\documentclass[a4paper,11pt]{article} 
\usepackage[english]{babel}
\usepackage[a4paper]{geometry}
\usepackage{amsmath}
\usepackage{amsthm}
\usepackage{amssymb}
\usepackage{color}
\usepackage{hyperref} 
\usepackage{enumerate}

\def\bR{\mathbb{R}}

\def\bT{\mathbb{T}}
\def\bN{\mathbb{N}}

\def\bZ{\mathbb{Z}}

\def\ph{\varphi}

\def\bC{\mathbb{C}}

\def\tbf{\textbf }
\def\t{\text }

\def\wh{\widehat}
\def\be{\begin{equation}}
\def\ee{\end{equation}}

\def\wh{\widehat}

\newtheorem{theorem}{Theorem}   
\newtheorem{cor}[theorem]{Corollary}
\newtheorem{prop}[theorem]{Proposition}
\newtheorem{lemma}[theorem]{Lemma}

\title{On the Leading Order Term of the Lattice \\ Yang-Mills Free Energy}
\author{Christian Brennecke\thanks{Institute for Applied Mathematics, University of Bonn, Endenicher Allee 60, 53115 Bonn, Germany}
}

\begin{document}

\maketitle

\begin{abstract}
In \cite{Cha1}, the leading order term of the free energy of $\t{U(N)}$ lattice Yang-Mills theory in $\Lambda_n=\{0,\ldots,n\}^d\subset \bZ^d$ was determined, for every $N\geq 1$ and $d\geq 2$. The formula is explicit apart from a contribution $K_d$ which corresponds to the limiting free energy of lattice Maxwell theory with boundary conditions induced by the axial gauge. By suitably adjusting the boundary conditions, we provide an equivalent characterization of $K_d$ that admits its explicit computation.  
\end{abstract}

\section{Introduction and Main Result} 

A major open problem in mathematical physics is to rigorously construct non-Abelian Euclidean Yang-Mills theories in $d\geq 3$ dimensions which might then be used to put the corresponding quantum Yang-Mills theories on firm ground. After great progress in the past (see e.g.\ \cite{Balaban1, Balaban2, Balaban3, Balaban4, Balaban5, Balaban6, Balaban7,Balaban8,Balaban9,Balaban10,Balaban11,Balaban12, Brydges1, Brydges2,Federbush6, GlimmJaffe, JaffeWitten, MagRivSen,OstSei, Seiler} as well as \cite{ChaReview} for a recent review that contains a thorough list of references including results on Abelian models and results in dimension $d=2$), the development of novel methods has lead to renewed interest in this and related questions, see e.g.\ \cite{Cha1, Cha2, ChaReview, Chev, Cao, Cha3, ChaCheHaiShen1, ShenZhuZhu1, CheShen1, CPS, CC1, CC2, ChaCheHaiShen2, Adhikari, SSZ, AdhikariCao, ShenZhuZhu2, Cha4, CheShen2, ChaYakir}. Major progress has in particular been obtained towards the construction of Euclidean Yang-Mills theories in $d=3$ including the rigorous construction of potential state spaces for such theories with compact gauge group $G$ on the unit torus $\bT^3$. This was recently obtained in \cite{ChaCheHaiShen2, CheShen2} based on the methods of stochastic quantization and in \cite{CC1, CC2} based on the regularization of Wilson loop observables through the Yang-Mills heat flow, an idea that goes back to \cite{CharGross}. In addition to the state space, \cite{CC1} provides a tightness criterion for constructing Euclidean Yang-Mills theories from approximate Yang-Mills theories (such as lattice Yang-Mills theories) with Gaussian free field like short distance behavior. Such behavior is expected based on perturbation theory and was rigorously established in dimension $d=2$ in \cite{Chev}. 

Motivated by these developments, we follow in this note \cite{Cha1} and consider Euclidean lattice Yang-Mills theories in $d\geq 2$ with gauge group $\t{U(N)}$, for $N\geq 1$. A basic important question in this setting is to determine the leading order term of the free energy as the volume tends to infinity. This was solved in \cite{Cha1} and our main result is related to it. Before stating this, let us introduce the precise setup.

Let $\Lambda \subset \bZ^d$ be a finite subset and let $G$ denote a closed subgroup of the group of $N\times N$ unitary matrices, $G\subset \t{U(N)}=\{ U\in \bC^{N\times N}: U^*U = \tbf 1_{\bC^N} \}$, with Lie algebra $\frak{g}$. We denote by $ E_\Lambda$ the set of oriented edges $ e=(x,y)\in E_\Lambda$, which are pairs of adjacent vertices $x,y \in \Lambda$ such that $x \prec y \,(=x+e_i$ for some $i\in [d]=\{1,\ldots,d\}$) in the lexicographic ordering. The state space of the theory is $G^{E_\Lambda}$, the set of maps from $E_\Lambda$ to $G$, called lattice gauge fields. If $U = (U_e)_{e\in E_\Lambda}\in G^{E_\Lambda}$ is a field configuration, $e = (x,y)\in E_\Lambda$ and $e'=(y,x)$ denotes the negatively oriented version of $e$, we set $U_{e'} = U_{e}^* = U_e^{-1}$. A plaquette $p$ is a square that is bounded by four edges in $E_\Lambda$. It can be characterized uniquely by the triple $p=(x,j,k)\in \Lambda\times [d]^2$, if its vertices are given by $ x, x+e_j, x+e_j+e_k, x+e_k \in\Lambda$ for some $1\leq j < k \leq d$. Here and in the following, $e_1, e_2,\ldots, e_d \in\bR^d$ denote the canonical basis vectors of $\bR^d$. Note that in this characterization, $x\in \Lambda_n$ is the smallest (with regards to the lexicographic ordering) vertex contained in $p$. The set of plaquettes in $\Lambda$ is denoted by $P_\Lambda$. For a configuration $U\in G^{E_\Lambda}$ and a plaquette $p=(x,j,k)\in P_\Lambda$, we set 
		\[\begin{split}
		U_p &= U_{(x,x+e_j)} U_{(x+e_j, x+e_j+e_k)} U_{(x+e_k, x+e_j+e_k)}^*U_{(x, x+e_k)}^* \\
		& = U_{(x,x+e_j)} U_{(x+e_j, x+e_j+e_k)} U_{(x+e_j+e_k, x+e_k)}U_{(x+e_k, x)} \in\t{U(N)}. \end{split}\]
Then, the Wilson action $S_\Lambda: G^{E_\Lambda}\to [0,\infty)$ is defined by 
		\be\label{eq:defWilson} S_\Lambda(U) = \sum_{p\in P_\Lambda} \text{Re} \big(\text{tr}\, (\tbf 1_{\bC^N}- U_p)\big)  =  \frac12 \sum_{p\in P_\Lambda} \| \tbf 1_{\bC^N}- U_p \|^2. \ee
Here, $ \|A\|^2 = \sum_{i,j=1}^N|a_{ij}|^2 $, for  $A=(a_{ij})_{i,j=1}^N\in\bC^{N\times N}$, equals the square of the Hilbert-Schmidt norm which is induced by the inner product $\langle A,B\rangle= \text{tr}\, A^*B$. Notice that the second equality in \eqref{eq:defWilson} uses that $U_p\in \t{U(N)}$ so that $\langle U_p, U_p\rangle= \langle \tbf 1_{\bC^N}, \tbf 1_{\bC^N}\rangle$.

Given the above notions, the Euclidean lattice (pure) Yang-Mills measure $\mu_{\Lambda,g}$ with coupling strength $g>0$ is the measure on $U^{E_\Lambda}$ defined by 
		\be\label{eq:defYMmeas} \mu_{\Lambda,g}(dU) = \frac{1}{Z_{\Lambda,g}} \exp\Big(-\frac1{g^2} S_\Lambda(U)\Big) \sigma_\Lambda(dU), \ee
where $\sigma_\Lambda(dU) = \prod_{e\in E_\Lambda} dU_e$ denotes the normalized product Haar measure on $G^{E_\Lambda}$. The normalizing constant $Z_{\Lambda,g} $ in \eqref{eq:defYMmeas} is the partition function of the model and it equals
		\[ Z_{\Lambda,g} = \int_{G^{E_\Lambda}}  \exp\Big(-\frac1{g^2} S_\Lambda(U)\Big)\,\sigma_\Lambda(dU). \]
The corresponding free energy per site $F_{\Lambda,g}$ is defined by 
		\[  F_{\Lambda,g} = \frac{\log Z_{\Lambda,g}}{|\Lambda|}. \]

From now on, we focus on the lattice $\Lambda = \Lambda_n = [0,n]^d\cap \bZ^d= \{0,\ldots,n\}^d$ so that $|\Lambda_n|=n^d$. For simplicity, we abbreviate $ E_n=E_{\Lambda_n}, P_{n} = P_{\Lambda_n}, \mu_{n,g}= \mu_{\Lambda_n,g}, Z_{\Lambda_n,g} = Z_{n,g}$ and $ F_{n,g}=F_{\Lambda_n,g} $. The main result of \cite{Cha1} is a derivation of the limit $\lim_{n\to\infty,g\to 0} F_{n,g}$, where $n\to\infty$ and $g\to0$ may vary independently. As pointed out in \cite{Cha1}, the small coupling limit $g\to 0$ is relevant for lattice Yang-Mills theory on the scaled lattice $ \Lambda_{n,\epsilon}=\{0, \epsilon,\ldots,\epsilon n\}^d\subset \epsilon \bZ^d$. By relabeling the edges, this is equivalent to \eqref{eq:defYMmeas} for $\Lambda=\Lambda_n$ with $g^{2}$ replaced by $g^{2}_\epsilon=g^{2}\epsilon^{4-d}$. Hence, as long as $d<4$, we have that $g_\epsilon^{2}\to 0$ as the lattice spacing $\epsilon \to 0$ tends to zero. A precise understanding of $Z_{n,g}$ and $\mu_{n,g}$ as $n\to\infty, g\to 0$ might therefore be helpful for the construction of continuum Yang-Mills measures, following the strategy as proposed in \cite{CC1}.

In \cite{Cha1}, the computation of the leading order term of $F_{n,g}$ is obtained by fixing the axial gauge and by approximating $\mu_{n,g}$ through an effective Gaussian theory sometimes referred to as lattice Maxwell theory. For a general introduction to lattice Maxwell theory in the context of renormalization of lattice gauge theories, see e.g.\ \cite[Section 22.6]{GlimmJaffe} and for a more detailed analysis, see e.g.\ \cite{Balaban1, Balaban2,Balaban3}. To state this more precisely, we need to introduce some further notation. Consider the maximal tree $T_n$ in $\Lambda_n$ with root $(0,\ldots, 0)\in \Lambda_n$ and edges in $E_n$. It contains an edge $e=(x,y)\in E_{n}$ if and only if its vertices $x,y\in \Lambda_n$ are of the form
		\be\label{eq:en0}  x= (x_1, x_2,\ldots, x_j, 0,\ldots,0),\;\;\; y=x+e_j= (x_1, x_2,\ldots, x_j +1, 0,\ldots,0), \ee
for some $j\in[d]$. This includes in particular every edge of the form $e=(x,x+e_d)\in E_n$. We denote the set of edges in $T_n$ by $E_{n}^0 = E_{T_n}\subset E_n$ and we identify
		\[ \bR^{E_n^1}\simeq \big\{ u \in  \bR^{E_{n}}: u_e =0, \,\forall e\in E_{n}^0\big\}   \;\text{ for }\; E_n^1 = E_n\setminus E_n^0. \] 
Loosely speaking, elements in $\bR^{E_n^1}$ correspond to a Lie algebra component of the gauge field connections in the axial gauge. For $e=(x,y)\in E_{n}$, we set $u_{e'}=-u_e$ if $e'=(y,x)$ denotes the negatively oriented version of $e\in E_n$ and for $p=(x,j,k)\in  P_{n}$, we set 
		\be\label{eq:defap}\begin{split} u_p &= u_{(x,x+e_j)} + u_{(x+e_j,x+e_j+e_k)} - u_{(x+e_k,x+e_j+e_k)}- u_{(x,x+e_k)}\\
		&= u_{(x,x+e_j)} + u_{(x+e_j,x+e_j+e_k)} + u_{(x+e_j+e_k,x+e_k)}+ u_{(x+e_k,x)}.   
		\end{split} \ee
Finally, we define the non-negative symmetric quadratic form $\Sigma_n: \bR^E_{n}\times  \bR^E_{n}\to \bR$ by
		\be\label{eq:defSigma}\Sigma_n(u,v) =\sum_{p\in P_{n}} u_p v_p\ee
and we denote its restriction to $\bR^{E_n^1}$ by $\Sigma_n^0 :\bR^{E_n^1}\times \bR^{E_n^1}\to \bR$. For simplicity of notation, we denote the associated symmetric matrices $\Sigma_n \in\bR^{|E_n|\times|E_n|}, \Sigma_n^0 \in \bR^{|E_n^1| \times |E_n^1|}$ again by $\Sigma_n$ and, respectively, $\Sigma_n^0$. Since $\Sigma_n\geq 0$ (and hence $\Sigma_n^0\geq 0$), it is clear that 
		\[- \log \det \Sigma_n^0 = -\sum_{j=1}^{|E_n^1|} \log  \lambda_j\big(\Sigma_n^0\big)= -\t{tr} \log \Sigma_n^0\in (-\infty,\infty]\] 
is well-defined. Here, $ \lambda_j(\Sigma_n^0) \geq 0$ denotes the $j$-th min-max value of $\Sigma_n^0$, for $1\leq j\leq |E_n^1|$, recalled in \eqref{eq:minmax} below. Since the restriction from $E_n$ to $E_n^1$ is based on the choice of the axial gauge, we refer in the sequel to the Gaussian measure on $\bR^{E_n^1} $ with covariance $\Sigma_n^0$ as lattice Maxwell theory in the axial gauge. The main result of \cite{Cha1} reads as follows. 

\begin{theorem}[{\cite[Theorem 2.1]{Cha1}}] \label{thm:Cha} Let $d\geq 2$, $N\geq 1$ and set $G=\emph{U(N)}$. Then the limit 
		\be \label{eq:Kd} K_d = \lim_{n\to\infty} \frac{-\emph{tr} \log \Sigma_n^0}{2n^d} \in\bR \ee
exists and we have that
		\be\label{eq:freee}  F_{n,g} =  \frac{|E_n^1|}{2n^{d}} N^2\log g^2  + (d-1)\log \frac{\prod_{j=1}^{N-1} j!}{(2\pi)^{N/2}} +N^2 K_d + o(1),  \ee
where $|E_n^1|=  (d-1)n^d-dn^{d-1}+1$ and where $\lim_{n\to\infty, g\to0} o(1)=0$.
\end{theorem}

The proof of Theorem \ref{thm:Cha} in \cite{Cha1} combines various tools from probability theory, lattice gauge theory and statistical mechanics. The interpretation of the terms on the right hand side in \eqref{eq:freee} is as follows. Locally, one can change integration over $|E_n^1|$ copies of $G$ with regards to Haar measure to integration over $|E_n^1|$ copies of $\frak{g}$ against a measure that is absolutely continuous with regards to Lebesgue measure. Its Radon-Nikodym derivative is taken into account by the second term on the right hand side in \eqref{eq:freee}, which is related to the volume of balls in $\text{U(N)}$ (see \cite[Section 6]{Cha1} for the details). What remains are, to leading order, $ \t{dim}(\frak{g}) = N^2$ independent Gaussian fields which are distributed according to lattice Maxwell theory, rescaled by the factor $-1/g^2$. A change of variables in $ \bR^{|E_n^1|N^2}$ then yields the first term on the right hand side and we are left with the free energy $ N^2 K_d$ of $N^2$ independent Gaussian fields with covariance matrix $\Sigma_n^{0}$. In particular, as pointed out in \cite[Section 2]{Cha1}, it is interesting to note that the constant $K_d$ in \eqref{eq:Kd} depends only on the dimension $d\geq 2$, but not on the specific gauge group $G$ and it was left open whether it can be determined explicitly. The purpose of this note is to address this question and to provide a formula for $K_d$. 

Let $ E_n, E_n^0$ and $E_n^1$ be defined as above and let $u\in\bR^{E_n} $. In the sequel, it is useful to identify $u\in\bR^{E_n} $ with a discrete one form. To this end, for simplicity of notation, it turns out convenient to embed $\bR^{E_n}\hookrightarrow \bR^{E_{\bZ^d}}$ by setting $u_e =0$ whenever $ e\in E_{\bZ^d}\setminus E_n$, a convention we adopt from now on. Notice that this implies in particular that $ u_{(x,x+e_d)} =0 $ for all $u\in\bR^{E_n^1}$ and $x\in \Lambda_n$, because either $(x,x+e_d)\in E_n^0$ or $(x,x+e_d)\not \in E_n$, for every $x\in \Lambda_n$. As a consequence, the map 
		\be\label{eq:defiso} \bR^{E_n} \ni  u= (u_{e})_{e\in E_n} \mapsto  w^{(u)} =\big(w_i^{(u)}\big)_{i=1}^{d}, \; \big(w_i^{(u)}\big)(x) = u_{(x,x+e_i)}, \forall x\in\Lambda_n, i\in [d], \ee
defines a linear isomorphism between $\bR^{E_n^1}\subset \bR^{E_n}$ and the space 
		\be \label{eq:space} \Omega_{n}^{1,\t{a}} = \Omega_{\Lambda_n}^{1,\t{a}} = \Big\{  w= (w_i)_{i=1}^{d} \in (\bR^{\Lambda_n})^{d}:   (w_i)_{|V_{n,i}^0} = 0,\,\forall  i\in [d]\Big\}, \ee
where
		\be \label{eq:bndcond}V_{n,i}^0= V_{\Lambda_n,i}^0 = \big\{x\in \Lambda_n: (x,x+e_i)\in E_n^0 \big\}\cup \big\{x\in \partial \Lambda_n: (x, x+e_i)\not \in  E_n \big\}. \ee
An element $ w= (w_i)_{i=1}^{d}\in \Omega_{n}^{1,\t{a}}$ corresponds to a discrete connection form on $\Lambda_n$ in the so called axial gauge. In the continuum, the corresponding gauge constraints for a connection one form $ A = \sum_{i=1}^d A_i \,\text{d}x_i \in \Omega^1(M, \mathfrak{g})$ on a base manifold $M$ read $ A_1(x_2 =x_3=\ldots=x_d=0)=0, A_2(x_3=\ldots=x_d=0)=0, \ldots, A_{d-1}( x_d=0)=0, A_d =0$, see e.g.\ \cite[Chapter 22]{GlimmJaffe}. In the discrete setting, the boundary condition on the $i$-th component, for $i\in[d]$, is encoded by $ (w_i)_{|V_{n,i}^0} = 0$, where $V_{n,i}^0$ is defined in \eqref{eq:bndcond}.

Since $w_d=0$ for $w\in \Omega_{n}^{1,\t{a}}$, we identify $\Omega_n^{1,\t a}$ in the sequel with a subspace
		\[\Omega_n^{1,\t a}\subset \bigoplus_{j=1}^{d-1} \ell^2(\Lambda_n) \hookrightarrow \bigoplus_{j=1}^{d} \ell^2(\bZ^d). \]
Here, $w\in \bigoplus_{j=1}^{d-1} \ell^2(\Lambda_n)$ is canonically identified with $w\in \bigoplus_{j=1}^{d} \ell^2(\bZ^d)$ by setting $w_{|\bZ^d \setminus \Lambda_n}=0$ and $w_d=0\in \ell^2(\bZ^d)$. The inner product in (real) $ \bigoplus_{j=1}^{d} \ell^2(\bZ^d)$ is 
		\[\langle w, w'\rangle =\sum_{i=1}^d \langle w_i, w_i' \rangle = \sum_{i=1}^d \sum_{x\in\bZ^d} w_i(x)w_i'(x), \;\forall w,w'\in \bigoplus_{j=1}^{d} \ell^2(\bZ^d). \]
The key observation that enables an explicit computation of $K_d$ from Theorem \ref{thm:Cha} is to identify the covariance $\Sigma_n$, defined in \eqref{eq:defSigma}, with a suitable lattice differential operator on $\bigoplus_{j=1}^{d} \ell^2(\bZ^d)$, restricted to $\Omega_{n}^{1,\t{a}}$. Such an identification was used quite recently in the massive setting (i.e.\ with $\Sigma_n$ replaced by $\Sigma_n + \varepsilon \tbf 1_{\bR^{E_n}}$ for $\varepsilon>0$) in \cite[Section 4.3]{Cha3} to derive the massive Proca field from $\t{SU}(2)$ lattice Yang-Mills-Higgs theory. The form of the differential operator is dictated by the perturbative emergence of lattice Maxwell theory from lattice Yang-Mills theory (see e.g.\ \cite[Chapter 22]{GlimmJaffe}). For the statement of our main result, define the operator $Q_d: \bigoplus_{j=1}^{d} \ell^2(\bZ^d)\to \bigoplus_{j=1}^{d} \ell^2(\bZ^d)$ by
		\be\label{eq:defQd} \big(Q_d w\big)_i =   -\Delta w_i - \sum_{j=1}^{d}  \partial_i \partial_j^* w_j , \;\forall i\in [d].   \ee
Then, $Q_d$ defines a non-negative, symmetric quadratic form in $ \bigoplus_{j=1}^{d} \ell^2(\bZ^d)$ such that
		\be\label{eq:defQdform} \langle w, Q_d w\rangle = \frac12 \sum_{i,j=1}^d \big\| \partial_i w_j-\partial_j w_i\big\|^2,\;\forall w\in  \bigoplus_{j=1}^{d} \ell^2(\bZ^d).\ee
Here, the operators $\partial_i, \partial_i^*$ and $ -\Delta = \sum_{i=1}^d \partial_i^*\partial_i$ refer to the standard lattice derivatives and, respectively, Laplacian. We recall their definitions in Section \ref{sec:LMT} below. For completeness, the simple proof of \eqref{eq:defQdform} is explained in Lemma \ref{lm:sigman2} below. By the same formula \eqref{eq:defQdform}, note that $Q_d = \text{d}^* \text{d}$ in terms of the discrete lattice exterior differential $\text{d}$ and its adjoint $\text{d}^*$. Given the above notation, the main result of this paper is as follows.

\begin{theorem}\label{thm:main} Let $d\geq 2$, let $Q_d: \bigoplus_{j=1}^{d} \ell^2(\bZ^d)\to \bigoplus_{j=1}^{d} \ell^2(\bZ^d)$ be as in \eqref{eq:defQd} and denote by $ \Pi_{\Omega_{n}^{1,\emph{a}}} $ the orthogonal projection of $\bigoplus_{j=1}^{d} \ell^2(\bZ^d)$ onto $\Omega_{n}^{1,\emph{a}}$. Then, we have that 
		\be \label{eq:Kdlattice} K_d =\lim_{n\to\infty} -\frac1{2n^d}\emph{tr}\big( \log \Pi_{\Omega_{n}^{1,\emph{a}}} Q_d\, \Pi_{\Omega_{n}^{1,\emph{a}}}\big) , \ee
where $\emph{tr}= \emph{tr}_{|\Omega_{n}^{1,\emph{a}}}$ denotes the trace over $ \Omega_{n}^{1,\emph{a}}$, and $K_d$ is given by the explicit formula
		\be\label{eq:K234}\begin{split}
		K_d 
		&=- \frac{d-2}2 \int_{[0,1]^d} \emph{d}x_1\ldots \emph{d}x_d\,  \log \sum_{k=1}^d 2\big(1-\cos(2\pi x_k)\big).
		\end{split}\ee
\end{theorem}

\noindent \textbf{Remarks:} 
\begin{enumerate}[1)]
\item  Notice that the constant 
		$$- \frac{1}2 \int_{[0,1]^d} \emph{d}x_1\ldots \emph{d}x_d\,  \log \sum_{k=1}^d 2\big(1-\cos(2\pi x_k)\big)$$ 
equals precisely the free energy of a scalar Gaussian free field (GFF) on $\bT^d$, interpreted as an appropriate limiting free energy of GFFs on the lattice. Starting from another common gauge, the so called Coulomb gauge $ \text{d}^* w = 0$ so that $Q_d w = ( \text{d}^*\text{d} + \text{d}\text{d}^*)w = (-\Delta)w$, it may at first seem surprising that $K_d$ is equal to the free energy of $d-2$ independent GFFs, instead of $d-1$ independent copies. This difference is due to the axial gauge and somewhat hidden in \eqref{eq:K234}. Indeed, the intuition behind \eqref{eq:K234} is as follows. The kernel of $Q_d$ is essentially given by the space of gradients. The axial gauge ensures the positivity of $Q_d$ and sets the $d$-th field component $w_d=0$ equal to zero so that, in the axial gauge, $Q_d$ acts on gradients like $\partial_d^*\partial_d$. The free energy of a Gaussian field on $\Lambda_n$ with this covariance equals 
		\[ - \frac{1}2 \int_0^1  \emph{d}x\,  \log 2 \big(1-\cos(2\pi x)\big) = 0.\]
For the $d-2$ remaining orthogonal field components, on the other hand, $Q_d$ acts like the free Laplacian $-\Delta$. Splitting the field into a gradient and $d-2$ orthogonal components, we thus arrive at \eqref{eq:K234}. 

\item To make the heuristic argument from the previous remark rigorous, we compare the spectrum of $\Pi_{\Omega_{n}^{1,\t{a}}}Q_d\, \Pi_{\Omega_{n}^{1,\t{a}}}$ with the spectrum of the related operator 
 		\[\Pi_{\Omega_{n}^{1,\t{p,+}}}Q_d^{\t{per}}\, \Pi_{\Omega_{n}^{1,\t{p,+}}},\] 
which corresponds to $Q_d$ with periodic boundary conditions above its zero ground state energy. Loosely speaking, our guiding principle is that $K_d$ equals a limiting free energy density which often does not depend on the specific boundary conditions and that $Q_d$ can be diagonalized explicitly in the periodic setting.  
\end{enumerate}

In Section \ref{sec:LMT}, we provide an elementary and self-contained proof of the reduction to the periodic case and we outline the computation that leads to \eqref{eq:K234}. A key ingredient from \cite{Cha1} that is used in our proof is that the axial gauge ensures the positivity of the lowest eigenvalue of $\Pi_{\Omega_{n}^{1,\t{a}}}Q_d\, \Pi_{\Omega_{n}^{1,\t{a}}}$. As a side remark, notice that given this input, our derivation of \eqref{eq:K234} provides an alternative existence proof of $K_d$ (cf.\ \cite[Section 15]{Cha1}).  

\section{Free Energy of Lattice Maxwell Theory}\label{sec:LMT}

In this section, we provide a detailed proof of Theorem \ref{thm:main}. We proceed in several steps. The first lemma identifies $\Sigma_n$, defined in \eqref{eq:defSigma}, with a suitable restriction of the operator $Q_d$ from Theorem \ref{thm:main}. Before stating this precisely, let us recall that  
		\be\label{eq:dd*}\begin{split} 
		(\partial_i \phi)(x) &= \phi(x+e_i)-\phi(x),  \hspace{0.5cm} (\partial_i^* \phi)(x) = \phi(x-e_i)-\phi(x)
		\end{split}\ee
for all $ \phi \in \ell^2(\bZ^d)$, $x\in \bZ^d$ and $i\in [d]$. Moreover, we recall that the discrete Laplacian $-\Delta$ in $\bZ^d$ is defined by 
		\[-\Delta = \sum_{i=1}^d \partial_i^*\partial_i = \sum_{i=1}^d \partial_i\partial_i^*\]
so that 
		\[ (-\Delta \phi)(x)  =  2d \phi(x) - \sum_{i=1}^d  \phi(x-e_i) - \sum_{i=1}^d\phi(x+e_i) \] 
for all $ \phi \in \ell^2(\bZ^d)$ and $x\in \bZ^d$. Notice that $[\partial_i, \partial_j^*] = [\partial_i, \partial_j]=0$ for all $i,j\in [d]$ and that $\langle \phi, (-\Delta) \phi\rangle = \|\nabla\phi\|^2\geq 0$ as well as $\langle \phi, \partial_i \psi\rangle = \langle \partial_i^*\phi,  \psi\rangle $ for all $\phi,\psi \in \ell^2(\bZ^d)$. Here, we set $[A,B] = AB-BA$ for linear operators $A,B: \ell^2(\bZ^d)\to \ell^2(\bZ^d)$ and $\nabla\phi = (\partial_1\phi,\ldots,\partial_d\phi)\in \bigoplus_{j=1}^d \ell^2(\bZ^d)$ denotes the discrete gradient of $\phi \in \ell^2(\bZ^d)$.

In the first lemma, we compute the matrix representation $\Sigma_n \in\bR^{|E_n|\times |E_n|}$ of the quadratic form $\Sigma_n$, defined in \eqref{eq:defSigma}, explicitly. Up to a few modifications related to the boundary edges in $\Lambda_n$, we argue as in the proof of \cite[Theorem 4.6]{Cha3}. Following the latter, let us recall the following terminology. Let $p=(x,j,k)\in P_{n}$ be a plaquette, for some $1\leq j<k\leq d$, so that $x\in\Lambda_n$ is its smallest vertex in lexicographic ordering. Let us abbreviate its edges by $e^{(1)}= (x,x+e_j), e^{(2)}=(x+e_j, x+e_j+e_k), e^{(3)} = (x+e_k, x+e_j+e_k), e^{(4)}=(x,x+e_k)\in E_n $ so that $e^{(2)} $ is the edge that touches $e^{(1)} $ at its right endpoint $x+e_j\in\Lambda_n$ while $e^{(4)} $ is the edge that touches $e^{(1)}$ at its left endpoint $x\in\Lambda_n$. We say that the edges $e^{(1)}, e^{(2)},e^{(3)},e^{(4)}$ of $p$ are neighbors and according to the signs in the plaquette's weight $u_p$ defined in \eqref{eq:defap}, i.e.
		\[u_p = u_{e^{(1)}}+u_{e^{(2)}}-u_{e^{(3)}}-u_{e^{(4)}},\] 
we say that the edges $e^{(j)}, e^{(k)}$ are positive neighbors if $\{j,k\}=\{1,2\}$ or $\{j,k\}=\{3,4\}$. In all other cases, we say that $e^{(j)}, e^{(k)}$ are negative neighbors. Given an edge $e\in E_n$, we denote the set of all its positive neighbors by $N_+(e)\subset E_n$ and the set of all of its negative neighbors by $N_-(e)\subset E_n$. 

In order to compute the matrix representation of $\Sigma_n$ from \eqref{eq:defSigma} explicitly, it turns out useful to decompose the set of edges $E_n$ as follows. First, viewing $\partial \Lambda_n = \partial [0,n]^d\cap \bZ^d$ as a subgraph of $\Lambda_n$, we split the set of edges  
		\[ E_n = E_{\Lambda_n} = E_{\Lambda_n}^\circ \cup E_{\partial\Lambda_n}\]
disjointly into its boundary edges $E_{\partial\Lambda_n} = \{e=(x,y)\in E_n: x,y\in\partial \Lambda_n\} $ and its complement $E_{\Lambda_n}^\circ = E_n\setminus E_{\partial\Lambda_n}$, whose elements we call interior edges of $\Lambda_n$.

In the next step, we repeat the decomposition into boundary and interior edges with $E_{\Lambda_n}$ replaced by $E_{\partial\Lambda_n}$. To this end, notice that 
		\[ \partial\Lambda_n = \bigcup_{j=1}^{2d} F_{\Lambda_n,j}  \]
is equal to the union over the $2d$ faces given by   
		\[\begin{split}
		F_{\Lambda_n,j} &= \big\{ x=(x_1,\ldots,x_d)\in\Lambda_n: x_j=0 \big\},\; F_{\Lambda_n,j+d} = \big\{ x=(x_1,\ldots,x_d)\in\Lambda_n: x_{j}=n \big\}, 
		\end{split} \]
for $ j\in[d]$. Each face $F_{\Lambda_n,k}$ is isomorphic to $[0,n]^{d-1}\cap \bZ^{d-1}$. Accordingly, we split the edges of $ F_{\Lambda_n,k} $ (identified with $[0,n]^{d-1}\cap \bZ^{d-1}$) disjointly into its interior and boundary edges, similarly as in the first step. The union over $j\in [2d]$ of all resulting interior face edges is denoted by $ E_{\partial\Lambda_n}^\circ $ and the union over all the boundary face edges is denoted by $E_{\partial^2\Lambda_n}$. Hence
		\[ E_{\partial\Lambda_n} = E_{\partial\Lambda_n}^\circ\cup E_{\partial^2\Lambda_n}.\]
Iterating this procedure, we conclude that $E_n$ splits into the disjoint union
		\be \label{eq:bndintedges}E_n = E_{\Lambda_n} = \bigcup_{j=0}^{d-1} E_{\partial^j\Lambda_n}^\circ,  \ee
where we set $ E_{\partial^0\Lambda_n}^\circ = E_{\Lambda_n}^\circ$, where $ \partial^k \Lambda_n$ is a union of $(d-k)$-dimensional faces contained in $\partial^{k-1}\Lambda_n$ each of which is isomorphic to $ [0,n]^{d-k}\cap \bZ^{d-k}$ (notice that $ \partial^k \Lambda_n$ consists of all points $x\in\Lambda_n$ for which $k$ coordinates are fixed and equal to $0$ or $n$) and where $E_{\partial^{d-1}\Lambda_n}^\circ = E_{\partial^{d-1}\Lambda_n}$, as $\partial^d\Lambda_n= \{0,n\}^d$ equals the set of corners of $\Lambda_n$, so that $ E_{\partial^{d}\Lambda_n}=\emptyset$.   

\begin{lemma}\label{lm:sigman1} Let $ \Sigma_n \in\bR^{|E_n|\times |E_n|}$ denote the symmetric matrix associated with the quadratic form defined in \eqref{eq:defSigma}. Denote its matrix entries by $(\Sigma_n)_{ee'} $ for $e,e'\in E_n$. Then
		\be\label{eq:Sigmanrep}  (\Sigma_n)_{ee'} = \begin{cases} 2(d-1) - k & \text{if } e=e', e\in (\partial^kE_{\Lambda_n})^\circ \t{ and } 0\leq k\leq d-1, \\  1 & \text{if $e,e'$ are positive neighbors}, \\ -1 & \text{if $e,e'$ are negative neighbors}, \\ 0 & \text{else.}
		\end{cases}\ee
\end{lemma}
\begin{proof} 
Using the shorthand notation $e^{(1)}= (x,x+e_j), e^{(2)}=(x+e_j, x+e_j+e_k), e^{(3)} = (x+e_k, x+e_j+e_k), e^{(4)}=(x,x+e_k)\in E_n $ for a plaquette $p=(x,j,k)$, $1\leq j<k\leq d$, as introduced above, we have that 
		\[ \Sigma_n(u,v)=\sum_{p=(x,j,k)\in P_n} u_pv_p, \]
where
		\be\label{eq:sigman1}\begin{split}
		u_pv_p&=   u_{e^{(1)}}v_{e^{(1)}} +u_{e^{(2)}}v_{e^{(2)}}+u_{e^{(3)}}v_{e^{(3)}}+u_{e^{(4)}}v_{e^{(4)}} + u_{e^{(1)}}v_{e^{(2)}}+ u_{e^{(2)}}v_{e^{(1)}}\\
		& \hspace{0.4cm}+u_{e^{(3)}}v_{e^{(4)}} +u_{e^{(4)}}v_{e^{(3)}}-u_{e^{(1)}}v_{e^{(3)}}-u_{e^{(3)}}v_{e^{(1)}}  -u_{e^{(1)}}v_{e^{(4)}}-u_{e^{(4)}}v_{e^{(1)}}\\
		&\hspace{0.4cm}- u_{e^{(2)}}v_{e^{(3)}}-u_{e^{(3)}}v_{e^{(2)}} -u_{e^{(2)}}v_{e^{(4)}}-u_{e^{(4)}}v_{e^{(2)}}.  
		\end{split}\ee
The identity \eqref{eq:sigman1} implies that $\Sigma_n(u,v)$ consists of a sum of terms of the form $u_e v_e, u_{e}v_{e'}$ and $-u_{e}v_{e'}$, for edges $e, e'\in E_n$ such that $e,e'$ are neighbors. Comparing this with 
		\be\label{eq:sigman2}\Sigma_n(u,v) = \langle u, \Sigma_n  v\rangle = \sum_{e,e'\in E_n} (\Sigma_n)_{ee'} u_ev_{e'},\ee
it implies that $(\Sigma_n)_{ee'}=0$ unless $e, e'\in E_n$ are neighbors. If $e,e'\in E_n$ are neighbors, on the other hand, they must be two edges of the same plaquette and there exists at most one such plaquette (for two distinct plaquettes share at most one edge). In this case, the sum in \eqref{eq:sigman2} contains exactly the two terms $(\Sigma_n)_{ee'}u_e v_{e'}$ and $(\Sigma_n)_{e'e}u_{e'} v_e$ in which $e,e'$ are paired.  
If $e,e'\in E_n$ are positive neighbors, we thus conclude from \eqref{eq:sigman1} that $(\Sigma_n)_{ee'}=(\Sigma_n)_{e'e}=1$ and if they are negative neighbors then $(\Sigma_n)_{ee'}=(\Sigma_n)_{e'e}=-1$.

It remains to determine the diagonal entries $(\Sigma_n)_{ee}$ of $\Sigma_n$, for all $e\in E_n$. According to \eqref{eq:sigman1}, this amounts to counting how many plaquettes $p\in P_n$ contain a given edge $e\in E_n$. Here, we consider several cases according to the decomposition of $E_n$ in \eqref{eq:bndintedges}. Assume first that $e=(x,y) \in E_{\Lambda_n}^\circ$ so that either $x\not\in\partial\Lambda_n$ or $y\not\in\partial\Lambda_n$. Then, if the edges of a plaquette $p=(x,j,k)\in P_n$, for some $1\leq j<k\leq d$, are denoted by $e^{(1)}, e^{(2)},e^{(3)}$ and $e^{(4)}$ as above, $e \in E_{\Lambda_n}^\circ$ can be equal to each one of them. Clearly, writing $e =(x,x+e_j)$ for some $1\leq j\leq d$, we find $d-j$ plaquettes that contain $e=e^{(1)}$, which are uniquely determined by $e$ and $(x,x+e_k)$ for $j<k\leq d$, and $j-1$ plaquettes with $e=e^{(4)}$, which are uniquely determined by $e$ and $(x,x+e_k)$ for $1\leq k<j$. Similarly, we conclude that there exist $j-1$ plaquettes for which $e=e^{(2)} $ and $ d-j$ plaquettes for which $e=e^{(3)} $. In summary, this shows that $ (\Sigma_n)_{ee}=2(d-1)$ for all $e\in E_{\Lambda_n}^\circ$. 

As a consequence of the previous step, we claim that  $ (\Sigma_n)_{ee}=2(d-1)-k$ for all $e\in E_{\partial^{k}\Lambda_n}^\circ$ and for all $0\leq k\leq d-1$. Indeed, if $e\in  E_{\partial^k\Lambda_n}^\circ$, it is an interior edge of a face isomorphic to $[0,n]^{d-k}\cap \bZ^{d-k}$. By the first step, there are exactly $ 2(d-k-1)$ plaquettes whose edges lie entirely in $\partial^k\Lambda_n\subset \Lambda_n$. To determine the remaining plaquettes that contain $e$, we may assume without loss of generality that $e=(x,y)$ for two adjacent vertices $x=(0,\ldots,0,x_{k+1} \ldots,x_d) \prec y=(0,\ldots,0,y_{k+1}, \ldots,y_d)\in \Lambda_n$. Then, there are exactly $k$ additional plaquettes $p_i\in P_n$ that contain $e$ and whose edges are equal to $ e=(x,y), (y,y+e_i), (y+e_i, x+e_i), (x,x+e_i)\in E_n$, for $1\leq i\leq k$. Thus, $ (\Sigma_n)_{ee}=2(d-k-1)+k =2(d-1)-k$ for all $e\in E_{\partial^{k}\Lambda_n}^\circ$ and $0\leq k\leq d-1$. We thus find that
		\[\begin{split}
		\Sigma_n(u,v) = \sum_{e,e'\in E_n: e\neq e'} \!\!(\Sigma_n)_{ee'}u_e v_{e'} + \sum_{k=0}^{d-1}\sum_{e\in E^\circ_{\partial^k\Lambda_n} } \!\!\! |\{ p\in P_n: e \text{ is edge of } p\}|\,u_ev_e =\langle u,\Sigma_n v\rangle	
		\end{split}\]
for all $u, v\in \bR^{E_n}$ and for $ \Sigma_n \in\bR^{|E_n|\times |E_n|}$ as defined in \eqref{eq:Sigmanrep}. 
\end{proof}

In the next step, we use Lemma \ref{lm:sigman1} to identify $\Sigma_n$ with the lattice differential operator $Q_d$ on $\bigoplus_{j=1}^{d} \ell^2(\bZ^d)$, defined in \eqref{eq:defQd}.

\begin{lemma}\label{lm:sigman2} Let $Q_d:\bigoplus_{j=1}^d\ell^2(\bZ^d)\to \bigoplus_{j=1}^d\ell^2(\bZ^d)$ be defined as in \eqref{eq:defQd}. Moreover, define the operator $R_d:\bigoplus_{j=1}^d\ell^2(\bZ^d)\to \bigoplus_{j=1}^d\ell^2(\bZ^d)$ by 
			\[ (R_d w)_i(x) = \begin{cases} k w_i(x) & \t{if } (x,x+e_i) \in E_{\partial^k \Lambda_n}^\circ,\, 1\leq k\leq d-1, \\ 0 & \t{else.} \end{cases} \]
Then, for every $u\in \bR^{E_n}$, we have that
		\[ \Sigma_n(u,u) =  \big\langle w^{(u)}, Q_d w^{(u)} \big\rangle - \big\langle w^{(u)}, R_d w^{(u)}\big\rangle,\]
where $w^{(u)}\in \bigoplus_{j=1}^d \ell^2(\Lambda_n)\hookrightarrow \bigoplus_{j=1}^d \ell^2(\Lambda_n)$ is defined as in \eqref{eq:defiso}. Moreover, we have that
		\[ \langle w, Q_dw\rangle =\frac12 \sum_{i,j=1}^d \big\| \partial_i w_j-\partial_j w_i\big\|^2, \;\forall w\in \bigoplus_{j=1}^d \ell^2(\bZ^d). \] 
\end{lemma}
\begin{proof} Let $u\in \bR^{E_n}$ and recall that by convention $u_e=0$ for all $e\not\in E_n$. We have 
		\[(\Sigma_n u)_e = \sum_{e'\in E_n} (\Sigma_n)_{ee'} u_e = (\Sigma_n)_{ee} u_e + \sum_{e'\in N_+(e)}  u_{e'} - \sum_{e'\in N_-(e)} u_{e'}.  \] 
To evaluate the r.h.s.\ further, consider first an edge $e=(x,x+e_i)\in E_{\Lambda_n}^\circ$, $x\in\Lambda_n,   i\in[d]$. Then, as already pointed out in \cite[Section 4.3]{Cha3}, it is straightforward to check that 
		\be\label{eq:nbrs}\begin{split}
		N_+(e) &= \bigcup_{j\in [d]: j\neq i}\big\{ (x+e_i, x+e_i+e_j) \big\} \cup \big\{ (x-e_j, x) \big\},\\ 
		N_-(e) &= \bigcup_{j\in [d]: j\neq i}\big\{ (x+e_j, x+e_i+e_j) \big\} \cup \big\{  (x-e_j, x+e_i-e_j)\big\}\\
		&\hspace{1.7cm}\cup  \big\{ (x, x+e_j) \big\} \cup\big\{ (x+e_i-e_j, x+e_i) \big\} 
		\end{split}\ee
with $N_+(e), N_-(e)\subset E_n $. Abbreviating $w^{(u)}= w =(w_1,\ldots,w_d) $, we get
		\[\begin{split}
		(\Sigma_n u)_e 
		& =  2(d-1) w_i(x) - \sum_{j\in[d]: j\neq i} \big( w_i(x+e_j) + w_i(x-e_j)\big)\\
		&\hspace{0.4cm} +\sum_{j\in[d]: j\neq i} \big(   w_j(x+e_i)- w_j(x) + w_j(x-e_j) - w_j(x-e_j+e_i)\big)\\
		& =  (-\Delta w_i)(x)+\sum_{j=1}^d \big(   w_j(x+e_i)- w_j(x) + w_j(x-e_j) - w_j(x-e_j+e_i)\big)\\
		& = (-\Delta w_i)(x) + \sum_{j=1}^d \big(  (\partial_iw_j)(x)-(\partial_iw_j)(x-e_j)\big)\\
		& = (-\Delta w_i)(x) - \sum_{j=1}^d \partial_i\partial_j^*w_j (x).
		\end{split}\]
For $ e\in E_{\partial^k\Lambda_n}^\circ \subset E_{\partial \Lambda_n}\subset E_n$, $1\leq k\leq d-1$, only the diagonal term $(\Sigma_n)_{ee}$ in the previous calculation needs to be modified. In fact, in this case the positive and negative neighbor sets $N_+(e)\subset E_n$ and $N_-(e)\subset E_n$ correspond to certain subsets of the two sets displayed in \eqref{eq:nbrs}, for some of the corresponding edges in \eqref{eq:nbrs} are not contained in $E_n$. If $ e'=(x',x'+e_j)\not \in E_n$, however, then $u_{e'}=0$ by definition and hence $ w_j^{(u)}(x')=0$ as well, by \eqref{eq:defiso}. This means that we can sum over the same neighbor sets displayed in \eqref{eq:nbrs}. Taking also the modification of $(\Sigma_n)_{ee}$ into account, we thus find  
		\[(\Sigma_n u)_e =  (-\Delta w_i)(x)  +\sum_{j=1}^d (-\partial_i\partial_j^*w_j  )(x) - k w_i(x)\]
for every $e=(x,x+e_i)\in E_{\partial^k\Lambda_n}^\circ, 0\leq k\leq d-1$. By definition of $R_d$, this shows that 
		\[\begin{split}
		\Sigma_n(u,u)  &= \sum_{i=1}^d\Big( \langle w_i, (-\Delta)w_i\rangle+  \sum_{j=1}^d\langle w_i, -\partial_i\partial_j^*  w_j\rangle\Big)   - \langle w, R_dw\rangle\\ 
		&=  \langle w, Q_d w\rangle - \langle w, R_dw\rangle 
		\end{split}\]
for every $u\in \bR^{E_n}$ and $w=w^{(u)}\in \bigoplus_{j=1}^d \ell^2(\Lambda_n)$. 

Finally, since the adjoint of $\partial_i$ is $\partial_i^*$, we have for every  $w\in \bigoplus_{j=1}^d \ell^2(\bZ^d)$ that
		\[\begin{split}
		\langle w, Q_dw\rangle &=  \sum_{i,j=1}^d\Big( \langle w_i, \partial_j^* \partial_j w_i\rangle + \langle w_i, (-\partial_j^*\partial_i)w_j\rangle\Big)\\
		& = \frac12 \sum_{i,j=1}^d  \Big( \|  \partial_j w_i \|^2 + \|  \partial_i w_j \|^2 - 2 \langle \partial_j w_i,  \partial_i w_j\rangle\Big) =\frac12 \sum_{i,j=1}^d \big\| \partial_i w_j-\partial_j w_i\big\|^2.
		\end{split}\]
This proves in particular the identity \eqref{eq:defQdform}.
\end{proof}

We remark that $R_d$, defined in Lemma \ref{lm:sigman2}, is non-negative in the sense of forms on $\bigoplus_{j=1}^d\ell^2(\bZ^d)$, which follows directly from its definition. Notice furthermore that\footnote{We denote constants that may depend on $d\geq2$, but not on $n\in\bN$ typically by $C, c$. In estimates that range over several lines, such constants may change from line to line.} 
		\[ 0 \leq \langle w, Q_dw\rangle \leq C \|w\|^2, \;\; 0 \leq \langle w, R_dw\rangle \leq C \|w\|^2, \;\forall w\in \bigoplus_{j=1}^d\ell^2(\bZ^d),\]
for some constant $C=C_d>0$. This implies that both $Q_d$ and $R_d$ are bounded operators on $\bigoplus_{j=1}^d\ell^2(\bZ^d)$ (and hence on $\bigoplus_{j=1}^d\ell^2(\Lambda_n)$, for all $n\in\bN$), i.e.\
		\be\label{eq:QdRdbnd} \| Q_d\|_{\t{op}} = \sup_{\|w\|=1} |\langle w, Q_dw\rangle| \leq C,\;\; \| R_d\|_{\t{op}}\leq C. \ee

Our next goal is to compare $Q_d$ as an operator on $\bigoplus_{j=1}^{d-1}\ell^2(\Lambda_n)$ with boundary conditions induced by $\Omega_n^{1,\t{a}}$ with $Q_d$ and suitably modified boundary conditions. We proceed in two main steps. First, we show that the perturbation $R_d$ from Lemma \ref{lm:sigman2} is negligible in view of the computation of $K_d$. Afterwards, we focus on the analysis of $Q_d$ on $\bigoplus_{j=1}^{d-1}\ell^2(\Lambda_n)$ and compare the boundary conditions induced by $\Omega_n^{1,\t{a}}$ with periodic boundary conditions in a slightly enlarged box. 

Let us first recall some basic facts about the spectrum of symmetric matrices. Let $k\geq 1$ and let $A \in \bR^{k\times k}$ be a symmetric matrix that is non-negative in the sense of forms. Its eigenvalues (counted with multiplicity) in increasing order are denoted by 
		\[0\leq \lambda_1(A)\leq \lambda_2(A)\leq \ldots\leq \lambda_k(A).\]
They are characterized by the min-max formula (see e.g.\ \cite[Section 1.3]{Tao}) 
		\be\label{eq:minmax}  \lambda_j(A) = \inf_{\substack{V\subset \bR^k,\\\t{dim}(V) = j}} \,\sup_{\substack{u\in \bR^k,\\\|u\|=1} } \langle u, Au\rangle,\;1\leq j \leq k,\ee
where the infimum is taken over all $j$-dimensional linear subspaces of $\bR^k$. Normalized eigenvectors to the smallest eigenvalue $\lambda_1(A)$ are called the ground states of $A$. Since $A \in \bR^{k\times k}$ extends canonically to a non-negative self-adjoint matrix $A \in \bC^{k\times k}$ with the same spectrum, note that the infimum in \eqref{eq:minmax} may also be taken over subspaces $V\subset \bC^k$.

A simple consequence of \eqref{eq:minmax} is that for an orthogonal projection $\Pi: \bR^k\to\bR^k$ with $ \t{dim}\,\t{ran}(\Pi) = l$, the eigenvalues of $\Pi A\Pi:\t{ran}(\Pi)\to \t{ran}(\Pi)$, viewed as a linear map from $\t{ran}(\Pi)\simeq \bR^l$ to itself, are related to those of $A$ by  
		\be\label{eq:projconseq1}      \lambda_j(A) \leq \lambda_j(\Pi A\Pi ) \leq \lambda_{j+k-l}(A), \;\forall \,1\leq j\leq l. \ee
The first inequality follows from the upper bound on the r.h.s.\ of \eqref{eq:minmax} that is obtained by taking the infimum over $j$-dimensional subspaces $V\subset \t{ran}(\Pi)$. The second inequality in \eqref{eq:projconseq1} follows from \eqref{eq:minmax} and the fact that for every $(j+k-l)$-dimensional vector space $V\subset \bR^k$, we have $\text{dim}\,\Pi(V)\geq j$. This is a consequence of the rank formula, because
		\[ \t{dim}(V) = \t{dim} \,\t{ker}(\Pi_{|V})+ \t{dim} \,\t{ran}(\Pi_{|V}) \leq \t{dim}\, \t{ker}(\Pi) + \text{dim} \,\Pi(V) = k-l + \text{dim} \,\Pi(V). \] 

Consider now $ \Sigma_n^0\in \bR^{|E_n^1|\times |E_n^1|}$. By \eqref{eq:space}, \eqref{eq:bndcond} and Lemma \ref{lm:sigman2}, we have that 
		\[ \lambda_j\big(\Sigma_n^0\big) = \lambda_j \big( \Pi_{\Omega_{n}^{1,\t{a}}} (Q_d- R_d)\Pi_{\Omega_{n}^{1,\emph{a}}} \big),\; \forall 1\leq j \leq k_{\t a}= |E_n^1|,  \] 
where $ \Pi_{\Omega_{n}^{1,\t{a}}}$ denotes the orthogonal projection from $\bigoplus_{j=1}^d\ell^2(\bZ^d)$ onto $ \Omega_{n}^{1,\t{a}} \subset \bigoplus_{j=1}^{d-1}\ell^2(\Lambda_n)$ and where $\Pi_{\Omega_{n}^{1,\t{a}}} (Q_d- R_d)\Pi_{\Omega_{n}^{1,\t{a}}}$ is viewed as a linear map from $\Omega_{n}^{1,\t{a}} \simeq \bR^{k_{\t a}}$ to itself. Hence
		\be \label{eq:ljequi} K_d = -\lim_{n\to\infty}\frac{1}{2n^d} \t{tr} \log \big(\Pi_{\Omega_{n}^{1,\t{a}}} (Q_d- R_d)\Pi_{\Omega_{n}^{1,\t{a}}}\big), \ee
where the trace $\t{tr} = \t{tr}_{|\Omega_{n}^{1,\emph{a}}}$ is understood as the trace over $\Omega_{n}^{1,\emph{a}}\subset \bigoplus_{j=1}^{d-1}\ell^2(\Lambda_n)$. Recall here that $w\in  \Omega_{n}^{1,\emph{a}}$ implies $w_d=0\in \ell^2(\Lambda_n)$, by the axial gauge.

An important fact proved in \cite[Lemma 13.1]{Cha1} and used repeatedly below is that the lowest eigenvalue $\lambda_1\big(\Sigma_n^0\big)$ of $\Sigma_n^0 $ is strictly positive in the sense that
		\be\label{eq:gap1} \lambda_1\big(\Sigma_n^0\big) = \lambda_1 \big( \Pi_{\Omega_{n}^{1,\emph{a}}} (Q_d- R_d)\Pi_{\Omega_{n}^{1,\t{a}}} \big)\geq \frac{C}{n^{d+2}} >0. \ee
On the other hand, from the remarks after Lemma \ref{lm:sigman2} and from the bound \eqref{eq:QdRdbnd}, we obtain 
		\be \label{eq:lupperbnds}  \max_{ j=1,\ldots, k_{\t a}} \lambda_j \big( \Pi_{\Omega_{n}^{1,\emph{a}}} (Q_d- R_d)\Pi_{\Omega_{n}^{1,\t{a}}} \big)\leq \max_{ j\in\bN} \lambda_j \big(  Q_d \big)\leq \|Q_d\|_{\t{op}} \leq C \ee
for some constant $C=C_d>0$ that depends on $d $, but that is independent of $n$. 

\begin{prop}\label{prop:removeRd} We have that	
		\be\label{eq:KdQd} K_d = - \lim_{n\to\infty}\frac{1}{2n^d} \emph{tr} \log \big(\Pi_{\Omega_{n}^{1,\emph{a}}} Q_d \Pi_{\Omega_{n}^{1,\emph{a}}}\big). \ee
Moreover, for every subspace $V\subset \Omega_{n}^{1,\emph{a}}$ such that $\emph{dim}\,V^\bot \leq Cn^{d-1}$, we have that 
		\be\label{eq:KdQdV}  \frac{1}{2n^d} \emph{tr} \log \big(\Pi_{\Omega_{n}^{1,\emph{a}}} Q_d \Pi_{\Omega_{n}^{1,\emph{a}}}\big) =  \frac{1}{2n^d} \emph{tr} \log \big(\Pi_{V} Q_d \Pi_{V}\big)+O(\log n/n), \ee
where on the right hand side $\Pi_V$ denotes the orthogonal projection from $\Omega_{n}^{1,\emph{a}}$ to $V$ and where the trace is understood as trace $\emph{tr}=\emph{tr}_{|V}$ over $V\subset \Omega_{n}^{1,\emph{a}}$.
\end{prop}
\begin{proof}
Since $ R_d$ is a diagonal matrix with non-zero entries only for a subset of points contained in $ \partial\Lambda_n$ with $|\partial\Lambda_n|=O(n^{d-1})$ (for each block associated to the $d-1$ sectors in $\bigoplus_{j=1}^{d-1} \ell^2(\Lambda_n)\simeq \bR^{(d-1)|\Lambda_n|}$), it is clear that for some $C>0$, we have that
		\[\t{dim}\, \t{ran}(\Pi_{\Omega_{n}^{1,\t{a}}} R_d\Pi_{\Omega_{n}^{1,\t{a}}}) = Cn^{d-1}.\]
Let $\Pi_{\bot}$ denote the orthogonal projection from $ \Omega_{n}^{1,\t{a} }$ onto the orthogonal complement 
		\[\t{ran}(\Pi_{\Omega_{n}^{1,\t{a}}} R_d \Pi_{\Omega_{n}^{1,\t{a}}})^\bot\subset \Omega_{n}^{1,\t{a} }\]
with $ \t{dim}\,\t{ran}(\Pi_\bot)=k_\bot = k_{\t{a}}- Cn^{d-1}$. Then, we get from \eqref{eq:minmax}, \eqref{eq:projconseq1}, \eqref{eq:gap1} and \eqref{eq:lupperbnds}
		\[\begin{split}	
		\frac{1}{2n^d} \t{tr} \log \big(\Pi_{\Omega_{n}^{1,\emph{a}}} Q_d \Pi_{\Omega_{n}^{1,\t{a}}}\big) &\geq \frac{1}{2n^d} \t{tr} \log \big(\Pi_{\Omega_{n}^{1,\t{a}}} (Q_d- R_d)\Pi_{\Omega_{n}^{1,\t{a}}}\big)\\
		& \geq \frac{1}{2n^d} \sum_{j=1+k_{\t{a}}-k_\bot}^{k_{\t{a}}}\log \lambda_j\big(\Pi_{\Omega_{n}^{1,\t{a}}} (Q_d- R_d)\Pi_{\Omega_{n}^{1,\t{a}}}\big) \\  
		&\hspace{0.5cm}+ \frac1{2n^d}\sum_{j=1}^{k_{\t{a}}-k_\bot}  \log \lambda_1\big(\Pi_{\Omega_{n}^{1,\t{a}}} (Q_d- R_d)\Pi_{\Omega_{n}^{1,\t{a}}}\big) \\
		& \geq \frac{1}{2n^d} \sum_{j=1}^{k_\bot}\log \lambda_{j } \big(\Pi_\bot\Pi_{\Omega_{n}^{1,\t{a}}} Q_d \Pi_{\Omega_{n}^{1,\t{a}}}\Pi_\bot\big) +O( \log n/n)\\
		&\geq \frac{1}{2n^d} \sum_{j=1}^{k_{\t{a}}}\log \lambda_{j } \big(\Pi_{\Omega_{n}^{1,\t{a}}} Q_d \Pi_{\Omega_{n}^{1,\t{a}}}\big) +O( \log n/n)\\
		&\hspace{0.5cm}- \frac{C}n\log\lambda_{k_{\t{a}}}\big(\Pi_{\Omega_{n}^{1,\t{a}}} Q_d \Pi_{\Omega_{n}^{1,\t{a}}}\big)\\
		& = \frac{1}{2n^d} \t{tr} \log \big(\Pi_{\Omega_{n}^{1,\emph{a}}} Q_d \Pi_{\Omega_{n}^{1,\t{a}}}\big) +O( \log n/n).
		\end{split}\]
Notice that in the third step, we used that $ \Pi_\bot\Pi_{\Omega_{n}^{1,\t{a}}} R_d \Pi_{\Omega_{n}^{1,\t{a}}}\Pi_\bot=0$, by definition of $\Pi_\bot$. Combining the first and last inequalities with \eqref{eq:ljequi}, this shows that
		\[K_d =-\lim_{n\to\infty}\frac{1}{2n^d} \t{tr} \log \big(\Pi_{\Omega_{n}^{1,\t{a}}} (Q_d- R_d)\Pi_{\Omega_{n}^{1,\t{a}}}\big)=- \lim_{n\to\infty}\frac{1}{2n^d} \t{tr} \log \big(\Pi_{\Omega_{n}^{1,\t{a}}} Q_d \Pi_{\Omega_{n}^{1,\t{a}}}\big). \]

The identity \eqref{eq:KdQdV} follows analogously. Setting $k_V = \t{dim}\,V =k_{\t{a}}-Cn^{d-1}$, we find 
		\[\begin{split}
		\frac{1}{2n^d} \t{tr} \log \big(\Pi_{\Omega_{n}^{1,\emph{a}}} Q_d \Pi_{\Omega_{n}^{1,\t{a}}}\big) & \geq \frac{1}{2n^d} \sum_{j=1+k_{\t{a}}-k_V}^{k_{\t{a}}}\log \lambda_j\big(\Pi_{\Omega_{n}^{1,\t{a}}} Q_d \Pi_{\Omega_{n}^{1,\t{a}}}\big)+O( \log n/n) \\
		&\geq \frac{1}{2n^d} \sum_{j=1}^{k_V}\log \lambda_{j } \big( \Pi_{V} Q_d \Pi_{V}\big) +O( \log n/n)\\
		& = \frac{1}{2n^d} \t{tr} \log \big(\Pi_{V} Q_d \Pi_{V}\big)+O( \log n/n)\\
		&\geq \frac{1}{2n^d} \sum_{j=1}^{k_{\t{a}}}\log \lambda_{j } \big(\Pi_{\Omega_{n}^{1,\t{a}}} Q_d \Pi_{\Omega_{n}^{1,\t{a}}}\big) +O( \log n/n)\\
		& = \frac{1}{2n^d} \t{tr} \log \big(\Pi_{\Omega_{n}^{1,\emph{a}}} Q_d \Pi_{\Omega_{n}^{1,\t{a}}}\big) +O( \log n/n),
		\end{split}\]
which proves \eqref{eq:KdQdV}.
\end{proof}

In the next step, we switch from the axial gauge boundary conditions induced by $\Pi_{\Omega_{n}^{1,\t{a}}}$ to periodic boundary conditions in a slightly enlarged box. Our goal is to express $K_d$ as the free energy of lattice Maxwell theory with periodic boundary conditions above its zero ground state energy. This has the advantage that $Q_d$, viewed as a form on the torus $\bigoplus_{j=1}^{d-1}\ell^2(\bZ^d/n\bZ^d)$, can be diagonalized explicitly by the discrete Fourier transform.

Recall from \eqref{eq:space} and \eqref{eq:bndcond} that
		\[\begin{split}
		\Omega_{n}^{1,\t{a}} &=  \Big\{  w= (w_i)_{i=1}^{d} \in (\bR^{\Lambda_n})^{d}:   (w_i)_{|V_{n,i}^0} = 0,\,\forall  i\in [d]\Big\},\\
		V_{n,i}^0& = \big\{x\in \Lambda_n: (x,x+e_i)\in E_n^0 \big\}\cup \big\{x\in \partial \Lambda_n: (x, x+e_i)\not \in  E_n \big\}.
		\end{split}\]
As pointed out earlier, $w\in \Omega_{n}^{1,\t{a}} $ implies that $w_d=0$, so that we identify $ \Omega_{n}^{1,\t{a}} \subset \bigoplus_{j=1}^{d-1}\ell^2(\Lambda_n)\hookrightarrow \bigoplus_{j=1}^{d}\ell^2(\bZ^d)$, because either $(x,x+e_d)\in E_n^0$ or $(x,x+e_d)\not \in E_n$, for every $x\in \Lambda_n$. Moreover, by $E_n^0$ in \eqref{eq:en0}, observe that for every $1\leq i\leq d-1$, we have 
		\[ V_{n,i}^0\subset\partial^{d-i}\Lambda_n \cup \big\{x\in \partial \Lambda_n: (x, x+e_i)\not \in  E_n \big\}\subset \partial \Lambda_n.\]

Now, let us also recall some notation related to the periodic setting. For every $n\in\bN$, we denote by $\bT^d_n$ the discrete torus $\bT_n^d = \bZ^d/n\bZ^d$ and we define the translation invariant operator $Q_d^{\t{per}}: \bigoplus_{j=1}^{d-1}\ell^2(\bT_n^d)\to \bigoplus_{j=1}^{d-1}\ell^2(\bT_n^d)$ as in \eqref{eq:defQd}, i.e.\
		\[ \big(Q_d^{\t{per}} w\big)_i =   -\Delta w_i -\sum_{j=1}^{d-1}  \partial_i \partial_j^* w_j , \;\forall w\in \bigoplus_{j=1}^{d-1}\ell^2(\bT_n^d), i\in [d-1]. \]
Here, the derivatives $\partial_i, \partial^*_j$, for $i,j\in [d]$, and $(-\Delta)$ are defined pointwise as in \eqref{eq:dd*}, but with $x\in\bZ^d$ in \eqref{eq:dd*} replaced by $x\in \bT_n^d$. We have in particular that   
		\be\label{eq:Qperqform} \langle w,Q_d^{\t{per}} w\rangle = \frac12\sum_{i,j=1}^{d-1} \big\|\partial_i w_j-\partial_jw_i\big\|^2 + \sum_{i=1}^{d-1} \big\|\partial_d w_i\big\|^2, \;\forall w\in\bigoplus_{j=1}^{d-1}\ell^2(\bT_n^d).  \ee
		
To compare $ Q_d^{\t{per}}$ with $\Pi_{\Omega_{n}^{1,\t{a}}} Q_d \Pi_{\Omega_{n}^{1,\t{a}}}$, we slightly enlarge the volume and consider $\bT_{n+5}^{d}$ whose vertex set we identify with 
		\[\Lambda_{n+4}'=\{-2,-1,0,\ldots,n+1, n+2\}^d\simeq \Lambda_{n+4}.\] 
Extending every $w\in \Omega_{n}^{1,\t{a}}\subset \bigoplus_{j=1}^{d-1}\ell^2(\Lambda_n)$ by $w_{|\Lambda'_{n+4}\setminus\Lambda_n}=0$, the space $\Omega_{n}^{1,\t{a}}$ embeds canonically into $\bigoplus_{j=1}^{d-1}\ell^2(\bT_{n+5}^{d})$ through the map
		\[ \Omega_{n}^{1,\t{a}}\ni w\mapsto \iota_{\t{per}}(w)\in \bigoplus_{j=1}^{d-1}\ell^2(\bT_{n+5}^{d}) ,\;\iota_{\t{per}}(w)\big(x+(n+4)q\big) = w(x),\,\forall x\in \Lambda'_{n+4}, q\in\bZ^d. \]
We denote the subspace $ \iota_{\t{per}}\big(\Omega_{n}^{1,\t{a}}\big) \subset \bigoplus_{j=1}^{d-1}\ell^2(\bT_{n+5}^{d})$ by 
		\be \label{eq:dirichper} \begin{split}
		\Omega_{n+5}^{1,\t{a,p}} &= \Big\{ \iota_{\t{per}}(w): w\in \Omega_{n}^{1,\t{a}}, w_{|\Lambda'_{n+4}\setminus\Lambda_n}=0\Big\} \subset  \bigoplus_{j=1}^{d-1}\ell^2(\bT_{n+5}^{d}). 
		\end{split}\ee
Notice that $ \iota_{\t{per}}: \Omega_{n}^{1,\t{a}}\to\Omega_{n+5}^{1,\t{a,p}}$ defines an isometric isomorphism, i.e.\
		\be\label{eq:iso} \begin{split}
		\langle w, w'\rangle = \sum_{i=1}^d \sum_{x\in \Lambda_n} w(x)w'(x) &=\sum_{i=1}^d \sum_{x\in  \Lambda_{n+4}'} w(x)w'(x)  \\
		& =\sum_{i=1}^d \sum_{x\in  \Lambda_{n+4}'} \iota_{\t{per}}(w)(x)\iota_{\t{per}}(w')(x) \\
		&=\langle  \iota_{\t{per}}(w),  \iota_{\t{per}}(w')\rangle 
		\end{split}\ee
for every $w, w'\in \Omega_{n}^{1,\t{a}} $, and furthermore that 
		\be \label{eq:QdQdp}\langle w, \Pi_{\Omega_{n}^{1,\t{a}}}  Q_d \Pi_{\Omega_{n}^{1,\t{a}}} w\rangle = \langle w,  Q_d  w\rangle =  \big\langle \iota_{\t{per}}(w), Q_d^{\t{per}} \iota_{\t{per}}(w)\big\rangle \ee
for all $w \in \Omega_{n}^{1,\t{a}}$. The identity \eqref{eq:QdQdp} follows from the fact that both $Q_d $ and $Q_d^{\text{per}}$ are second order, nearest-neighbor difference operators and from $w_{|\Lambda'_{n+4}\setminus\Lambda_n}=0$ so that  
		\[ (Q_d w)(x) = \big(Q_d^{\t{per}} \iota_{\t{per}}(w)\big)(x),\;\forall x\in \Lambda_{n+4}',\] 
pointwise (where the left hand side is as in \eqref{eq:defQd}). In particular, this implies \eqref{eq:QdQdp}. 

Denoting by $\Pi_{\Omega_{n+5}^{1,\t{a,p}}}$ the orthogonal projection of $ \bigoplus_{j=1}^{d-1}\ell^2(\bT_{n+5}^{d})$ to $\Omega_{n+5}^{1,\t{a,p}}$ and viewing 
		\[\Omega_{n+5}^{1,\t{a,p}}Q_d^{\t{per}}\Omega_{n+5}^{1,\t{a,p}}:\text{ran}(\Omega_{n+5}^{1,\t{a,p}})\to \text{ran}(\Omega_{n+5}^{1,\t{a,p}}) \] 
as a linear operator from $\text{ran}(\Omega_{n+5}^{1,\t{a,p}})  \simeq \bR^{k_\t{a}}$  to itself, where $k_\t{a}=\t{dim}\ \t{ran}( \Pi_{\Omega_{n}^{1,\t{a}}})$, the previous observations and Proposition \ref{prop:removeRd} imply that
		\be\begin{split} \label{eq:linkper}  
		K_d &= - \lim_{n\to\infty} \frac{1}{2n^d}  \t{tr} \log \big(\Pi_{\Omega_{n}^{1,\t{a}}} Q_d  \Pi_{\Omega_{n}^{1,\t{a}}}\big)\\
		&= - \lim_{n\to\infty}\Big(1+\frac5n\Big)^d\frac{1}{2(n+5)^d}  \t{tr} \log \big(\Pi_{\Omega_{n+5}^{1,\t{a,p}}} Q_d^{\t{per}} \Pi_{\Omega_{n+5}^{1,\t{a,p}}}\big)\\
		&= - \lim_{n\to\infty}\frac{1}{2n^d} \t{tr} \log \big(\Pi_{\Omega_{n}^{1,\t{a,p}}} Q_d^{\t{per}} \Pi_{\Omega_{n}^{1,\t{a,p}}}\big).
		\end{split} \ee		
Here, the trace $ \t{tr}= \t{tr}_{|\Omega_{n}^{1,\t{a,p}}}$ in the last line is over $\Omega_{n}^{1,\t{a,p}}\subset \bigoplus_{j=1}^{d-1}\ell^2(\bT_{n})$. 

Similarly, if $V\subset \Omega_{n}^{1,\t{a,p}}$ is a subspace so that $V^\bot\subset \Omega_{n}^{1,\t{a,p}}$ satisfies $\t{dim}\,V^\bot\leq Cn^{d-1}$, the isometric property of $ \iota_{\t{per}}: \Omega_{n-5}^{1,\t{a}}\to\Omega_{n}^{1,\t{a,p}}$ and Proposition \ref{prop:removeRd} imply that 
		\be \label{eq:linkper2}  
		K_d = - \lim_{n\to\infty}\frac{1}{2n^d} \t{tr} \log \big(\Pi_{V} Q_d^{\t{per}} \Pi_{V}\big).
		\ee
Let us record furthermore that 
		\be \label{eq:gap3} \lambda_1\big(\Pi_{\Omega_{n}^{1,\t{a,p}}} Q_d^{\t{per}} \Pi_{\Omega_{n}^{1,\t{a,p}}}\big) = \lambda_1\big(\Pi_{\Omega_{n-5}^{1,\t{a}}} Q_d \Pi_{\Omega_{n-5}^{1,\t{a}}}\big) \geq \frac{C}{(n-5)^{d+2}}\geq \frac{C}{n^{d+2}}>0  \ee
which follows from \eqref{eq:gap1} and \eqref{eq:iso}, as well as 
		\[ \max_{j=1,\ldots,k_0} \lambda_j\big(\Pi_{\Omega_{n}^{1,\t{a,p}}} Q_d^{\t{per}} \Pi_{\Omega_{n}^{1,\t{a,p}}}\big)\leq \| Q_d^{\t{per}}\|_{\t{op}}\leq C,  \]
which can be proved similarly like \eqref{eq:QdRdbnd}. 

Our final characterization of $K_d$ is based on a comparison of $\Pi_{\Omega_{n}^{1,\t{a,p}}} Q_d^{\t{per}} \Pi_{\Omega_{n}^{1,\t{a,p}}}$ with $Q_d^{\t{per}}$ projected onto the orthogonal complement of the space of its ground states. Set 
		\[\begin{split}
		\Omega_{n}^{1,\t{p},+} =   \ker \big(Q_d^{\t{per}}\big)^\bot \subset \bigoplus_{j=1}^{d-1}\ell^2(\bT_n^d)
		\end{split}\]
and let $ \Pi_{\Omega_{n}^{1,\t{p},+}}$ denote the orthogonal projection of $\bigoplus_{j=1}^{d-1}\ell^2(\bT_n^d) $ onto $\Omega_{n}^{1,\t{p},+} $.

\begin{prop} \label{prop:Kdfinal} Consider $Q_d^{\emph{per}}:\bigoplus_{j=1}^{d-1}\ell^2(\bT_n^d)\to \bigoplus_{j=1}^{d-1}\ell^2(\bT_n^d)$. Then 
		\be  \label{eq:propfin1} \emph{dim} \ker \big(Q_d^{\emph{per}}\big)\leq Cn^{d-1}\ee
and we have that $\emph{spec} \big(Q_d^{\emph{per}}\big) = (\epsilon_p)_{p\in\Gamma_n^*}$, where 
		\be  \label{eq:propfin2}    \epsilon_p  = 2 \sum_{k=1}^{d} \big(1-\cos(2\pi p_k)\big). \ee
As a consequence, we get
		\be  \label{eq:propfin3}  \lambda_1\big( \Pi_{\Omega_{n}^{1,\emph{p},+}} Q_d^{\emph{per}}\Pi_{\Omega_{n}^{1,\emph{p},+}}\big) \geq \frac{C}{n^2}.\ee
		
Moreover, for every subspace $V\subset  \Omega_{n}^{1,\emph{p},+}$ with $ \emph{dim}\, V^\bot \leq Cn^{d-1}$, we have that
		\be  \label{eq:propfin4}   \frac{1}{2n^d} \emph{tr} \log \big(\Pi_V Q_d^{\emph{per}} \Pi_V \big) = \frac{1}{2n^d} \emph{tr} \log \big(\Pi_{\Omega_{n}^{1,\emph{p},+}} Q_d^{\emph{per}} \Pi_{\Omega_{n}^{1,\emph{p},+}} \big) + O(\log n/n), \ee
where $\Pi_V$ denotes the orthogonal projection of $\Omega_{n}^{1,\emph{p},+}$ onto $V$, so that in particular
		\be  \label{eq:propfin5}  K_d = - \lim_{n\to\infty}\frac{1}{2n^d} \emph{tr} \log \big(\Pi_{\Omega_{n}^{1,\emph{p},+}} Q_d^{\emph{per}} \Pi_{\Omega_{n}^{1,\emph{p},+}} \big).\ee
\end{prop}
\begin{proof} As mentioned after \eqref{eq:minmax}, we may consider $Q_d^{\t{per}}$ without loss of generality as a self-adjoint operator on the complex\footnote{We use the convention that $\langle\cdot,\cdot\rangle$ is conjugate linear in its first slot.} Hilbert space $\bigoplus_{j=1}^{d-1}\ell^2(\bT_n^d)$. This does neither change the spectrum nor the multiplicity of the eigenvalues of $Q_d^{\t{per}}$.

So, consider the plane wave basis $ (\ph_p)_{p\in \Gamma_{n}^*}$, defined by 
		\[ \ph_p(x) = n^{-\frac d2}e^{2\pi i px}, \,\forall x\in \bT_n^d, \; p=(p_1,\ldots, p_d)\in \Gamma_{n}^* = [0,n)^d \cap \frac{1}n\bZ^d. \]
Then, $(\ph_p)_{p\in \Gamma_{n}^*} $ is an orthonormal basis of $ \ell^2(\bT_n^d) \simeq \ell^2(\Gamma_n^*)$ and satisfies 		
		\be \label{eq:diag} 
		\begin{split}
		\partial_k \ph_p &= 2i e^{\pi ip_k} \sin(\pi p_k) \ph_p, \hspace{0.5cm}\partial_k^* \ph_p = -2ie^{-\pi ip_k} \sin(\pi p_k) \ph_p ,\\
		\partial_k^*\partial_k \ph_p &= \big(2 -2\cos(2\pi p_k)\big)\ph_p, \hspace{0.5cm} -\Delta  \ph_p  = 2 \sum_{k=1}^d \big(1-\cos(2\pi p_k)\big)\ph_p  ,\\
		\end{split}
		\ee
for every $j,k\in [d]$ and $p\in \Gamma_{n}^*$. Denoting by $ \widehat w \in \ell^2(\Gamma_{n}^*)$ the Fourier transform of $w\in \ell^2(\bT_n^d)$, defined so that
		\[\widehat w(p) = \langle \ph_p,w\rangle = \sum_{x\in\Lambda_{n-1}} \overline \ph_p(x) w(x) \;\text{ so that }\; w = \sum_{p\in\Gamma_{n}^*} \widehat w(p) \ph_p,\]
the identity \eqref{eq:Qperqform} reads in Fourier space 
		\be\label{eq:Qperqform2} \begin{split}
		\langle w,Q_d^{\t{per}} w\rangle &= \frac12\sum_{k,l=1}^{d-1} \big\|\partial_k w_l-\partial_lw_k\big\|^2 + \sum_{k=1}^{d-1} \big\|\partial_d w_k\big\|^2 \\
		& =  2 \sum_{k,l=1}^{d-1} \sum_{p\in\Gamma_{n}^*}  |e^{\pi ip_k} \sin(\pi p_k)\wh w_l(p)- e^{\pi ip_l}\sin(\pi p_l)\wh w_k(p)|^2 \\
		&\hspace{0.5cm}+ 2\sum_{k=1}^{d-1} \sum_{p\in\Gamma_{n}^*}  \big(1-\cos(2\pi p_d)\big) |\wh w_k(p)|^2. 
		\end{split} \ee
Here, we used Plancherel's theorem and \eqref{eq:diag}. 

Let us now determine $ \text{ker} (Q_d^{\t{per}})$ and $\text{spec}(Q_d^{\t{per}})$. By \eqref{eq:diag}, it is clear that for every $p\in\Gamma_n^*$, $ Q_d^{\t{per}}$ leaves the $(d-1)$-dimensional subspace
		\[ V_p=  \big\{ v \ph_p:v\in \bC^{d-1}\big\} \subset \bigoplus_{k=1}^{d-1}\ell^2(\bT_n^d)\] 
invariant. The spaces $V_p$, for $p\in\Gamma_n^*$, are mutually orthogonal, so it suffices to determine the kernel and spectrum of $Q_d^{\t{per}}$ restricted to $V_p$, for each $p\in\Gamma_n^*$. 

We consider several cases and start with $p=0\in\Gamma_n^*$. Then, clearly $V_p\subset \text{ker} (Q_d^{\t{per}})$. If $p\in \Gamma_n^*\setminus\{0\}$ is such that $ p=(0,\ldots,0,p_d)$ for some $p_d\neq 0$, on the other hand, then $V_p \subset  \big(\ker  Q_d^{\t{per}}\big)^\bot$, because for every $ v\in \bC^{d-1}$ we have that 
		\[ \big( Q_d^{\t{per}} v\ph_p\big)_i =  -\Delta v_i  \ph_p -\sum_{j=1}^{d-1}  \partial_i \partial_j^* v_j \ph_p =  2\big(1-\cos(2\pi p_d)\big)v_i\ph_p, \;\forall 1\leq i\leq d-1.   \]
In other words, $ V_p $ is an eigenspace of $Q_d^{\t{per}}$ to the eigenvalue $\epsilon_p =2\big(1-\cos(2\pi p_d)\big)>0$. 

Finally, let $p\in\Gamma_n^*\setminus\{0\}$ be such that $p_k\neq 0$ for some $ k\in [d-1]$. Then, if
		\[v\ph_p\in V_p\cap \ker Q_d^{\t{per}}= V_p\cap \ker \sqrt{Q_d^{\t{per}} },\] 
the identity \eqref{eq:Qperqform2} implies that
		\[\begin{split} v \ph_p &= \frac{v_k } { e^{\pi ip_k}\sin(\pi p_k)}\big( e^{\pi ip_1}\sin(\pi p_1),e^{\pi ip_2}\sin(\pi p_2), \ldots, e^{\pi ip_{d-1}}\sin(\pi p_{d-1})\big) \ph_p \\
		&= \frac{  v_k}{ 2i e^{\pi i p_k}\sin(\pi p_k)}  ( \partial_1\ph_p, \partial_2 \ph_p, \ldots, \partial_{d-1}\ph_p) \in \bigoplus_{j=1}^{d-1}\ell^2(\bT_n^d) 
		\end{split}\]
and that 
		\[ 2\big(1-\cos(2\pi p_d)\big) \sum_{j=1}^{d-1}|v_j|^2 =2\big(1-\cos(2\pi p_d)\big) |v|^2= 0.  \]
This means that $V_p \cap \ker  Q_d=\{0\} $ if $p_d\neq 0$ and if $p_d=0$, then 
		\[V_p\cap \ker Q_d^{\t{per}} = \t{span} \big( \psi_p^{(d-1)} \big) \;\;\text{ for }\; \;\psi_p^{(d-1)}=( \partial_1\ph_p, \partial_2 \ph_p, \ldots, \partial_{d-1}\ph_p)\]
is one-dimensional. Combined with the previous cases, note in particular that
		\[  \text{dim} \ker \big(Q_d^{\t{per}}\big) \leq \t{dim}\big(\t{span}\big((\ph_p)_{p=(p_1,\ldots,p_{d-1},0)\in\Gamma_n^*}\big) = \big|\big\{p\in\Gamma_{n}^*:p_d=0 \big\}\big|\leq C n^{d-1}, \]
which proves \eqref{eq:propfin1}. To determine the spectrum of $Q_d^{\t{per}}$ in $V_p$ for $p\in\Gamma_n^*\setminus\{0\}$ such that $p_k\neq 0$ for some $ k\in [d-1]$, consider the linearly independent vectors 
		\[ \psi_p^{(j)} = (0, \ldots, \partial_k^* \ph_p,0\ldots, 0, -\partial_j^*\ph_p,0, \ldots,0) \in \bigoplus_{k=1}^{d-1}\ell^2(\bT_n^d), \; 1\leq j\leq d-2.\]
Here, $ \partial_k^* \ph_p$ is in the $j$-th slot of $\psi_p^{(j)}$ and $-\partial_j^*\ph_p$ is in the $k$-th slot. Then we have that
		\[\begin{split}
		\big(Q_d^{\t{per}} \psi_p^{(d-1)}\big)_i = -\Delta  \partial_i \ph_p  -\sum_{l=1}^{d-1}  \partial_i \partial_l^*  \partial_l \ph_p = \partial_{d}^*  \partial_d (\partial_i\ph_p) = 2\big(1-\cos(2\pi p_d)\big)(\partial_i\ph_p)
		\end{split}\]
so that $(Q_d^{\t{per}} \psi_p^{(d-1)}) =2\big(1-\cos(2\pi p_d)\big) \psi_p^{(d-1)}$. For $j\in [d-2]$, we get on the other hand
		\[\begin{split}
		\big(Q_d^{\t{per}} \psi_p^{(j)}\big)_i &=  \big(-\Delta \psi_p^{(j)}\big)_i  - \partial_i \partial_j^* \big(\psi_p^{(j)}\big)_{j} - \partial_i \partial_k^* \big(\psi_p^{(j)}\big)_{k}  \\
		&=  \big(-\Delta \psi_p^{(j)}\big)_i -  \partial_i \partial_j^*\partial_k^* \ph_p + \partial_i \partial_k^* \partial_j^*\ph_p\\
		&= \big(-\Delta \psi_p^{(j)}\big)_i,
		\end{split}\]
so that $ (Q_d^{\t{per}} \psi_p^{(j)} ) = 2\sum_{l=1}^{d}\big(1-\cos(2\pi p_l)\big) \psi_p^{(j)}= \epsilon_p \psi_p^{(j)} $ and $\langle \psi_p^{(d-1)}, \psi_p^{(j)}\rangle=0$.  

In summary, the previous steps show that $\text{spec}(Q_d^{\t{per}}) = (\epsilon_p)_{p\in\Gamma_n^*} $, and in particular
		\[ \lambda_1\big( \Pi_{\Omega_{n}^{1,\t{p},+}} Q_d^{\t{per}}\Pi_{\Omega_{n}^{1,\t{p},+}}\big) \geq 2 \big(1-\cos(2\pi /n )\big)\geq \frac{C}{n^2},\]
which concludes both \eqref{eq:propfin2} and \eqref{eq:propfin3}.

Let us now switch to the proof of \eqref{eq:propfin4} and let us assume that $V\subset  \Omega_{n}^{1,\t{p},+}$ is a linear subspace such that $\t{dim}\, V^\bot \leq Cn^{d-1}$. By applying \eqref{eq:minmax}, \eqref{eq:projconseq1}, \eqref{eq:propfin3} and the boundedness of $ Q_d^{\t{per}}$ in $\bigoplus_{j=1}^{d-1} \ell^2(\bT_n^d)$, we get 
		\[\begin{split}
		\frac{1}{2n^d} \t{tr} \log \big(\Pi_{\Omega_{n}^{1,\t{p},+}} Q_d^{\t{per}} \Pi_{\Omega_{n}^{1,\t{p},+}} \big)& = \frac{1}{2n^d}\sum_{j=1}^{k_+} \log \lambda_j\big(\Pi_{\Omega_{n}^{1,\t{p},+}} Q_d^{\t{per}} \Pi_{\Omega_{n}^{1,\t{p},+}} \big)\\
		&\geq \frac{1}{2n^d}\sum_{j=1+k_+-k_V}^{k_+}\!\! \!\!\log \lambda_j\big(\Pi_{\Omega_{n}^{1,\t{p},+}} Q_d^{\t{per}} \Pi_{\Omega_{n}^{1,\t{p},+}} \big) - O(\log n/n)\\
		&\geq\frac{1}{2n^d}\sum_{j=1 }^{k_V} \log \lambda_j\big(\Pi_{V} Q_d^{\t{per}} \Pi_{V} \big) - O(\log n/n)\\
		& \geq \frac{1}{2n^d} \t{tr} \log \big(\Pi_{\Omega_{n}^{1,\t{p},+}} Q_d^{\t{per}} \Pi_{\Omega_{n}^{1,\t{p},+}} \big)- O(\log n/n),
		\end{split}\]
where we set $k_+ = \text{dim}\, \Omega_{n}^{1,\t{p},+}$ and $k_V=\text{dim}\, V$, such that $ k_+-k_V = O(n^{d-1})$. Hence
		\be \label{eq:subred}  \frac{1}{2n^d} \t{tr} \log \big(\Pi_V Q_d^{\t{per}} \Pi_V \big) = \frac{1}{2n^d} \t{tr} \log \big(\Pi_{\Omega_{n}^{1,\t{p},+}} Q_d^{\t{per}} \Pi_{\Omega_{n}^{1,\t{p},+}} \big) + O(\log n/n). \ee
		
Finally, we characterize $K_d$ in terms of $Q_d^{\t{per}}$, which uses \eqref{eq:linkper}, \eqref{eq:linkper2} and \eqref{eq:subred}. First of all, it is clear that $\big(\Omega_{n}^{1,\t{a,p}}\big)^\bot \subset \bigoplus_{j=1}^{d-1}\ell^2(\bT_n^d) $ satisfies
		\[ \t{dim}\,\big(\Omega_{n}^{1,\t{a,p}}\big)^\bot  \leq Cn^{d-1},\] 	 	
because $ \Omega_{n}^{1,\t{a,p}}\subset \bigoplus_{j=1}^{d-1}\ell^2(\bT_n^d) $ in \eqref{eq:dirichper} is characterized by $O(n^{d-1})$ zero constraints. Since $\t{dim} \ker \big(Q_d^{\t{per}}\big)\leq Cn^{d-1}$ as well, we also have that
		\[\text{dim} \,  \big(\Omega_{n}^{1,\t{a,p}}\cap  \Omega_{n}^{1,\t{p},+}\big)^\bot\leq Cn^{d-1}. \]
This follows e.g.\ from the dimension formula
		\[\begin{split} 
		&\text{dim} \, \big(\Omega_{n}^{1,\t{a,p}}\cap  \Omega_{n}^{1,\t{p},+}\big)+\text{dim} \,  \big(\Omega_{n}^{1,\t{a,p}} +\Omega_{n}^{1,\t{p},+}\big) = \text{dim} \,  \big(\Omega_{n}^{1,\t{a,p}} \big)+\text{dim} \,  \big(  \Omega_{n}^{1,\t{p},+}\big)
		\end{split}\]
which, combined with $\text{dim} \,  \big(\Omega_{n}^{1,\t{a,p}} +\Omega_{n}^{1,\t{p},+}\big)\leq \text{dim}\,\bigoplus_{j=1}^{d-1}\ell^2(\bT_n^d)$, implies that 
		\[\text{dim} \, \big(\Omega_{n}^{1,\t{a,p}}\cap  \Omega_{n}^{1,\t{p},+}\big)=\text{dim}\,\bigoplus_{j=1}^{d-1}\ell^2(\bT_n^d)-\text{dim} \,  \big(\Omega_{n}^{1,\t{a,p}}\cap  \Omega_{n}^{1,\t{p},+}\big)^\bot \geq \text{dim}\,\bigoplus_{j=1}^{d-1}\ell^2(\bT_n^d) - C n^{d-1}. \]
We can thus apply \eqref{eq:subred} to $V=\Omega_{n}^{1,\t{a,p}}\cap  \Omega_{n}^{1,\t{p},+}\subset \Omega_{n}^{1,\t{p},+} $. Since also $V\subset \Omega_{n}^{1,\t{a,p}}$, we can also apply \eqref{eq:linkper2}. Combined with \eqref{eq:linkper}, this proves that 
		\[\begin{split} K_d  = - \lim_{n\to\infty}  \frac{1}{2n^d} \t{tr} \log \big(\Pi_{\Omega_{n}^{1,\t{a,p}}} Q_d^{\t{per}} \Pi_{\Omega_{n}^{1,\t{a,p}}} \big)&=  -\lim_{n\to\infty} \frac{1}{2n^d} \t{tr} \log \big(\Pi_V Q_d^{\t{per}} \Pi_V \big)  \\
		&= -\lim_{n\to\infty} \frac{1}{2n^d} \t{tr} \log \big(\Pi_{\Omega_{n}^{1,\t{p},+}} Q_d^{\t{per}} \Pi_{\Omega_{n}^{1,\t{p},+}} \big).
		\end{split}\]
This implies \eqref{eq:propfin5} and concludes the proof of the proposition.
\end{proof}
Based on Proposition \ref{prop:Kdfinal}, we are finally prepared to compute $K_d$. 

\begin{cor}\label{lm:Kd}
Let $d\geq 2$. Then, we have that
		\[\begin{split}
		K_d 
		&=- \frac{d-2}2 \int_{[0,1]^d} \emph{d}x_1\ldots \emph{d}x_d\,  \log \sum_{k=1}^d 2\big(1-\cos(2\pi x_k)\big).
		\end{split}\]
\end{cor}
\begin{proof} We use the same notation as in the proof of Proposition \ref{prop:Kdfinal} and denote by $\Pi_{V_p}$, for $p\in \Gamma_n^*$, the orthogonal projection of $\Omega_{n}^{1,\t{p},+}  $ to $V_p$. From the proof of Proposition \ref{prop:Kdfinal}, we recall in particular that for every $ p=(p_1,\ldots,p_d)\in\Gamma_n^* $ with $p_1,\ldots,p_d\neq 0$, we have
		\[  \lambda_1\big( \Pi_{V_p}Q_d^{\t{per}}\Pi_{V_p} \big) =2\big(1-\cos(2\pi p_d)\big) < \lambda_k \big( \Pi_{V_p}Q_d^{\t{per}}\Pi_{V_p} \big) = \epsilon_p,\; \forall \,2\leq k\leq d-1.  \]
Combined with \eqref{eq:propfin5} and the mutual orthogonality of the spaces $V_p $, we readily conclude 
		\be\label{eq:Kdfinal}\begin{split}
		K_d 
		& = - \lim_{n\to\infty}\frac{1}{2n^d}  \sum_{\substack{ p=(p_1,\ldots,p_d) \in\Gamma_{n}^*:\\ p_1,\ldots, p_d\neq 0 }}\t{tr} \log \big(\Pi_{V_p} Q_d^{\t{per}} \Pi_{V_p} \big)+ \lim_{n\to\infty}\frac{1}{2n^d}O(n^{d-1}\log n) \\
		& =- - \frac12 \lim_{n\to\infty}  \sum_{p_1, \ldots,p_d=1}^{n-1}  \frac{1}{n^d } \log  2\big(1-\cos(2\pi p_d/n)\big)   \\
		&\hspace{0.5cm}  - \frac{d-2}2 \lim_{n\to\infty}   \sum_{p_1, \ldots,p_d=1}^{n-1}  \frac{1}{n^d }   \log \sum_{k=1}^{d} 2\big(1-\cos(2\pi p_k/n)\big)   \\
		&=  - \frac{1}2 \int_0^1  \emph{d}x\,  \log  2\big(1-\cos(2\pi x)\big)\\
		&\hspace{0.4cm}- \frac{d-2}2 \int_{[0,1]^d} \emph{d}x_1\ldots \emph{d}x_d\,  \log \sum_{k=1}^d 2\big(1-\cos(2\pi x_k)\big)\\
		&=  - \frac{d-2}2 \int_{[0,1]^d} \emph{d}x_1\ldots \emph{d}x_d\,  \log \sum_{k=1}^d 2\big(1-\cos(2\pi x_k)\big),
		\end{split}\ee 
where, in the last step, we used the elementary identity $\int_0^1  \emph{d}x\, \log  2\big(1-\cos(2\pi x)\big)=0$.
\end{proof}
		
\begin{proof}[Proof of Theorem \ref{thm:main}.] Proposition \ref{prop:removeRd} proves the characterization \eqref{eq:Kdlattice} of $K_d$ and from Corollary \ref{lm:Kd} we conclude \eqref{eq:K234}.
\end{proof}
\vspace{0.2cm}
\noindent\textbf{Acknowledgements.} C.\ B.\ thanks I.\ Chevyrev for helpful comments and two anonymous referees for several helpful remarks improving a previous version of this manuscript, in particular a simplification of \eqref{eq:K234} and its connection to the GFF as summarized in \tbf{Remark} 1) after Theorem \ref{thm:main}. C.\ B.\ acknowledges support by the Deutsche Forschungsgemeinschaft (DFG, German Research Foundation) under Germany’s Excellence Strategy – GZ 2047/1, Projekt-ID 390685813. 
\pagebreak 

\vspace{0.4cm}
\noindent {\Large \textbf{Data Availability}}
\vspace{0.4cm}

\noindent Data sharing is not applicable as no new data were created or analyzed in this study.

\vspace{0.4cm}
\noindent {\Large \textbf{Conflict of Interests}}
\vspace{0.4cm}

\noindent The author has no conflict of interest to declare that is relevant to the article's content.




\begin{thebibliography}{55}

\bibitem{Adhikari} A. Adhikari. Wilson loop expectations for non-Abelian finite gauge fields coupled to a Higgs boson at low and high disorder. \emph{Comm. Math. Phys.} \textbf{405} (2024).

\bibitem{AdhikariCao} A. Adhikari, S. Cao. Correlation decay for finite lattice gauge theories at weak coupling. \emph{Ann. Probab.} \textbf{53} (1), pp. 140-174 (2025).

\bibitem{Balaban1}  T. Ba\l aban. Regularity and decay of lattice Green’s functions. \emph{Comm. Math. Phys.} \textbf{89}, pp. 571-597 (1983).

\bibitem{Balaban2} T. Ba\l aban. Propagators and renormalization transformations for lattice gauge theories I. \emph{Comm. Math. Phys.} \textbf{95}, pp. 17-40 (1984).

\bibitem{Balaban3} T. Ba\l aban. Propagators and renormalization transformations for lattice gauge theories II. \emph{Comm. Math. Phys.} \textbf{96}, pp. 223-250 (1984).


\bibitem{Balaban4} T. Ba\l aban. Averaging operations for lattice gauge theories. \emph{Comm. Math. Phys.} \textbf{98}, pp. 17-51 (1985).

\bibitem{Balaban5} T. Ba\l aban. Spaces of regular gauge field configurations on a lattice and gauge fixing conditions. \emph{Comm. Math. Phys.} \textbf{99}, pp. 75-102 (1985).

\bibitem{Balaban6} T. Ba\l aban. Propagators for lattice gauge theories in a background field. \emph{Comm. Math. Phys.} \textbf{99}, pp. 389-434 (1985).

\bibitem{Balaban7} T. Ba\l aban. Ultraviolet stability of three-dimensional lattice pure gauge field theories. \emph{Comm. Math.Phys.} \textbf{102}, 255-275 (1985).

\bibitem{Balaban8} T. Ba\l aban. The variational problem and background fields in renormalization group method for lattice gauge theories. \emph{Comm. Math. Phys.} \textbf{102}, 277-309 (1985).

\bibitem{Balaban9} T. Ba\l aban. Renormalization group approach to lattice gauge field theories I. Generation of effective actions in a small field approximation and a coupling constant renormalization in four dimensions. \emph{Comm. Math. Phys.} \textbf{109} 249–301 (1987).

\bibitem{Balaban10} T. Ba\l aban. Convergent renormalization expansions for lattice gauge theories. \emph{Comm. Math. Phys.} \textbf{119}, 243-285 (1988).

\bibitem{Balaban11} T. Ba\l aban. Large field renormalization I. The basic step of the R operation. \emph{Comm. Math. Phys.} \textbf{122}, 175-202 (1989).

\bibitem{Balaban12} T. Ba\l aban. Large field renormalization II. Localization, exponentiation, and bounds for the R operation. \emph{Comm. Math. Phys.} \textbf{122}, 355-392 (1989).

\bibitem{Brydges1} D. Brydges, J. Fr\"ohlich, E. Seiler. On the construction of quantized gauge fields I. General results. \emph{Ann. Physics}, 121 no. 1-2, 227-284 (1979).

\bibitem{Brydges2} D. Brydges, J. Fr\"ohlich, E. Seiler. Construction of quantised gauge fields II. Convergence of the lattice approximation. \emph{Comm. Math. Phys.} \textbf{71}, no. 2, 159-205 (1980).


\bibitem{Cao} S. Cao. Wilson Loop Expectations in Lattice Gauge Theories with Finite Gauge Groups. \emph{Comm.\ Math.\ Phys.} \textbf{380} (2020).

\bibitem{CC1} S. Cao, S. Chatterjee. A State Space for 3D Euclidean Yang–Mills Theories. \emph{Comm.\ Math.\ Phys.} \textbf{304} (2024).

\bibitem{CC2} S. Cao, S. Chatterjee. The Yang-Mills heat flow with random distributional initial data. \emph{Commun.\ Partial Differ.\ Equ.} \textbf{48}, Issue 2 (2023).

\bibitem{CPS} S. Cao, M. Park, S. Sheffield. Random surfaces and lattice Yang-Mills. \emph{Communications of the American Mathematical Society} \textbf{5}, Issue 14, (2025).

\bibitem{ChaCheHaiShen1} A. Chandra, I. Chevyrev, M. Hairer, H. Shen. Langevin dynamic for the 2D Yang–Mills measure. Publications mathématiques de l’IHÉS. 136, pp. 1-147 (2022).

\bibitem{ChaCheHaiShen2} A. Chandra, I. Chevyrev, M. Hairer, H. Shen. Stochastic quantisation of Yang–Mills–Higgs in 3D. \emph{Invent. math.} \tbf{237}, pp. 541-696 (2024).

\bibitem{CheShen1} I. Chevyrev, H. Shen. Invariant measure and universality of the 2D Yang-Mills Langevin dynamic. \emph{Commun. Pure Appl. Math.}, e70043 (2026). https://doi.org/10.1002/cpa.70043

\bibitem{CheShen2}  I. Chevyrev, H. Shen. Uniqueness of gauge covariant renormalisation of stochastic 3D Yang-Mills-Higgs. \emph{Arch. Rational Mech. Anal.} \textbf{250}, 11 (2026).

\bibitem{CharGross} N. Charalambous, L. Gross The Yang-Mills Heat Semigroup on Three-Manifolds with Boundary. \emph{Commun. Math. Phys.} \textbf{317} (2013).

\bibitem{Cha1} S. Chatterjee. The leading term of the Yang–Mills free energy. \emph{J.\ Funct.\ Anal.} \tbf{271}, Issue 10 (2016).

\bibitem{ChaReview} S. Chatterjee. Yang-Mills for Probabilists. In: P. Friz, W. König, C. Mukherjee, S. Olla (eds). Probability and Analysis in Interacting Physical Systems. VAR75 2016. \emph{Springer Proceedings in Mathematics \& Statistics} \textbf{283}, Springer, Cham (2019).

\bibitem{Cha2} S. Chatterjee. Rigorous solution of strongly coupled SO(N) lattice gauge theory in the large N limit. \emph{Comm. Math. Phys.} \textbf{366}, pp. 203-268 (2019). 

\bibitem{Cha3} S. Chatterjee. A probabilistic mechanism for quark confinement. \emph{Comm. Math. Phys.} \textbf{385}, pp. 1007-1039 (2021).

\bibitem{Cha4} S. Chatterjee. A scaling limit of SU(2) lattice Yang–Mills–Higgs theory. Preprint: arXiv:2401.10507.

\bibitem{ChaYakir} S. Chatterjee, O. Yakir. Correlation decay for U(1) lattice Higgs theory: the case of small mass. Preprint: arXiv:2509.19176.

\bibitem{Chev} I. Chevyrev. Yang–Mills Measure on the Two-Dimensional Torus as a Random Distribution. \emph{Commun. Math. Phys.} \textbf{372}, 1027-1058 (2019).






\bibitem{Federbush6} P. Federbush. A phase cell approach to Yang-Mills theory. VI. Nonabelian lattice-continuum duality. \emph{Ann. Inst. H. Poincar\'e Phys. Th\'eor.} \textbf{47}, no. 1, 17-23 (1987b).

\bibitem{GlimmJaffe} J. Glimm, A. Jaffe. Quantum Physics. A Functional Integral Point of View. 2nd Edition. \emph{Springer New York} (1987).

\bibitem{JaffeWitten} A. Jaffe, E. Witten. Quantum Yang–Mills theory.  \emph{The millennium prize problems}, 129-152, Clay Math. Inst., Cambridge, MA (2006).

\bibitem{MagRivSen} J. Magnen, V. Rivasseau, R. S\'en\'eor. Construction of YM4 with an infrared cutoff. \emph{Comm. Math. Phys.} \textbf{155}, no. 2, 325-383 (1993).

\bibitem{OstSei} K. Osterwalder, E. Seiler. Gauge field theories on a lattice. \emph{Ann. Physics} \textbf{110}, no. 2, 440-471 (1978).




\bibitem{Seiler} E. Seiler. Gauge Theories as a Problem of Constructive Quantum Field Theory and Statistical Mechanics. In: J. Ehlers, K. Hepp, R. Kippenhahn, H. A. Weidenmüller, J. Zittartz (eds) Lecture Notes in Physics \textbf{159}, Springer Berlin Heidelberg (1982).

\bibitem{SSZ} H. Shen, S. A. Smith, R. Zhu. A new derivation of the finite N master loop equation for lattice Yang-Mills. \emph{Electron. J. Probab.} \textbf{29}, pp. 1-18 (2024).

\bibitem{ShenZhuZhu1} H. Shen, R. Zhu, X. Zhu. A stochastic analysis approach to lattice Yang-Mills at strong coupling. \emph{Comm. Math. Phys.} \textbf{400}, pp. 805-851 (2023). 

\bibitem{ShenZhuZhu2} H. Shen, R. Zhu, X. Zhu. Langevin dynamics of lattice Yang-Mills-Higgs and applications. Preprint: arXiv:2401.13299.



\bibitem{Tao} T. Tao. Topics in Random Matrix Theory. \emph{Graduate Studies in Mathematics} \tbf{132}, American Mathematical Society (2012).






\end{thebibliography}
\end{document}